\documentclass[final]{dmtcs-episciences}

\usepackage[utf8]{inputenc}

\usepackage[T1]{fontenc}
\usepackage{amsmath}

\usepackage{thm-restate}
\newtheorem{theorem}{Theorem}
\newtheorem{corollary}[theorem]{Corollary}
\newtheorem{lemma}[theorem]{Lemma}
\newtheorem{claim}[theorem]{Claim}

\title{Token Swapping on Trees\thanks{This work was begun at the University of Waterloo and was partially supported by the Natural Sciences and Engineering Council of Canada (NSERC).}
}

\author[Biniaz, Jain, Lubiw, Mas\'arov\'a, Miltzow, Mondal, Naredla, Tkadlec, Turcotte]{Ahmad Biniaz\affiliationmark{1}
\and Kshitij Jain\affiliationmark{2}
\and Anna Lubiw\affiliationmark{3}\\
\and Zuzana Mas\'arov\'a\affiliationmark{4}
\and Tillmann Miltzow\affiliationmark{5}
\and Debajyoti Mondal\affiliationmark{6}\\
\and Anurag Murty Naredla\affiliationmark{3}
\and Josef Tkadlec\affiliationmark{7}
\and Alexi Turcotte\affiliationmark{8}
}

\affiliation{
University of Windsor, Windsor, ON, Canada\\
Google Waterloo, Waterloo, ON, Canada\\
University of Waterloo, Waterloo, ON, Canada\\
IST Austria, Klosterneuburg, Austria\\
Utrecht University, Utrecht, Netherlands\\
University of Saskatchewan, Saskatoon, SK, Canada\\
Harvard University, Cambridge, MA, USA\\
Northeastern University, Boston, MA, USA
}


\begin{document}

\received{2021-08-23}
\revised{2022-09-19}
\accepted{2022-10-15}
\publicationdetails{24}{2022}{2}{9}{8383}

\maketitle

\begin{abstract}
 
The input to the token swapping problem is a graph with vertices $v_1, v_2, \ldots, v_n$, and $n$ tokens with labels $1, 2, \ldots, n$, one on each vertex. The goal is to get token $i$ to vertex $v_i$ for all $i= 1, \ldots, n$ using a minimum number of \emph{swaps}, where a swap exchanges the tokens on the endpoints of an edge. We present some results about token swapping on a tree, also known as ``sorting with a transposition tree'': 
\begin{enumerate}
\item An optimum swap sequence may need to perform a swap on a leaf vertex that has the correct token (a ``happy leaf''), disproving a conjecture of Vaughan.
\item Any algorithm  that fixes happy leaves---as all known approximation algorithms for the problem do---has approximation factor at least $4/3$.  Furthermore, the two best-known 2-approximation algorithms have approximation factor exactly 2.  
\item A generalized problem---weighted coloured token swapping---is NP-complete on trees, even when they are restricted to be subdivided stars, but solvable in polynomial time on paths and stars.  
In this version, tokens and vertices have colours, and colours have weights.  The goal is to get every token to a vertex of the same colour, and the cost of a swap is the sum of the weights of the two tokens involved.
\end{enumerate}
\end{abstract} 

\section{Introduction}
\label{sec:Introduction}

Suppose we wish to sort a list of numbers and the only allowable operation is to swap two adjacent elements of the list. It is well known that the number of swaps required is equal to the number of {\em inversions} in the list, i.e., the number of pairs that are out of order. Many other problems of sorting with a restricted set of operations have been studied, for example, pancake sorting, where the elementary operation is to flip a prefix of the list; finding the minimum number of pancake flips for a given list was recently proved NP-complete~\cite{bulteau2015pancake}.

A much more general problem arises when we are given a set of generators of a permutation group, and asked to express a given permutation $\pi$ in terms of those generators. Although there is a polynomial time algorithm to test if a permutation can be generated, finding a minimum length generating sequence was proved PSPACE-complete in 1985~\cite{jerrum1985complexity}.

This paper is about a problem, known recently in the computer science community as \emph{token swapping}, that is intermediate between sorting a list by swaps and general permutation generation. The input is a graph with $n$ vertices $v_1, \ldots, v_n$.  There are $n$ tokens, labelled $1, 2, \ldots, n$, and one token is placed on each vertex. The goal is to ``sort'' the tokens, which means getting token $i$ on vertex $v_i$, for all $i=1, \ldots, n$. The only allowable operation is to \emph{swap} the tokens at the endpoints of an edge, i.e., if $e = (v_i,v_j)$ is an edge of the graph and token $k$ is at $v_i$ and token $l$ is at $v_j$, then we can move token $k$ to $v_j$ and token $l$ to $v_i$. See Figure~\ref{fig:tree-example}. The \emph{token swapping problem} is to find the minimum number of swaps to sort the tokens. In terms of permutation groups, the generators are the transpositions determined by the graph edges, and the permutation is 
$\pi(i) = j$ if token $j$ is initially at vertex $v_i$;   
we want a minimum length generating sequence for the permutation.

\begin{figure}[h]
    \centering
    \includegraphics[width=.7\textwidth]{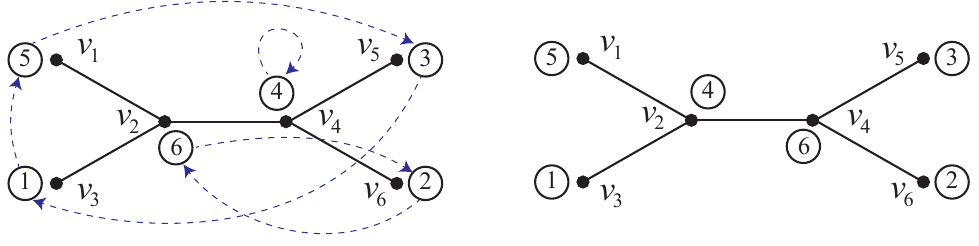}
    \caption{An example of the token swapping problem. Left: a tree of 6 vertices and an initial placement of tokens (in circles) on the vertices.   Blue dashed arrows indicate where each token should go.  Token 4 is home.  The corresponding permutation is $(1\  5\  3) (2\  6) (4)$.  Right: the effect of swapping tokens 4 and 6. Now token 6 is closer to its destination but token 4 is further from its destination. 
    One swap sequence that sorts the tokens to their destinations is (4\ 6), (6\ 2), (2\ 4), (3\ 4), (3\ 2), (3\ 1), (1\ 5), (5\ 2), (5\ 4). This sequence has 9 swaps, but there is a swap sequence of length 7.}
    \label{fig:tree-example}
\end{figure}

Our emphasis is on computing the number of swaps, and the actual swap sequence, needed for a given placement of tokens on a graph. Another interesting problem is to bound the worst-case number of swaps for a given graph, where the worst case is taken over all possible token placements.
This can be formulated in terms of the {\em Cayley graph}, which has a vertex for each possible assignment of tokens to vertices, and an edge when the two token assignments differ by one swap. The worst case number of swaps is equal to the diameter of the Cayley graph. The Cayley graph applies more generally for any permutation group given by a set of generators, where the generators define the edges of the Cayley graph.  This is discussed in more detail in Section~\ref{sec:perm}. 

In the special case when the graph is a path, the token swapping problem 
is precisely the classic problem of sorting a list using adjacent swaps, see Knuth~\cite{knuth1997art3}. Our paper is about token swapping on a tree.  This problem is also known as ``sorting with a transposition tree,'' and is of great interest in the area of sorting networks because the Cayley graph of a star (a tree with one non-leaf) is a good alternative to a hypercube. Akers and Krishnamurthy~\cite{akers1989group} first introduced this idea in 1989, and their paper has been cited more than 1600 times according to Google scholar. 
 
There are three different polynomial-time 2-approximation algorithms for token swapping on trees~\cite{akers1989group,vaughan1995algorithm,yamanaka2015swapping,miltzow2016approximation} (details below). Very recently (after the present work) token swapping on a tree was proved NP-complete~\cite{aichholzer2021hardness}.

Token swapping on general graphs has been studied by  
different research communities in math, computer science, and engineering, often unaware of each others' work.
We survey all the results we know of in Section~\ref{sec:background} below.

The problem of token swapping on graphs was proved NP-complete and in fact APX- hard~\cite{miltzow2016approximation},  and further hardness results have appeared since then~\cite{bonnet2017complexity}. There are polynomial time algorithms for paths, cliques~\cite{cayley}, 
cycles~\cite{jerrum1985complexity}, and stars~\cite{akers1989group,portier-vaughan1990star,pak1999reduced}, and some other special cases, as discussed in more detail below.

The token swapping problem has been generalized in several ways. In \emph{weighted token swapping} each token $i$ has a positive weight $w(i)$  and the cost of swapping token $k$ and token $l$ is $w(k) + w(l)$.  The goal is to sort the tokens while minimizing the sum of the costs of the swaps. In \emph{coloured token swapping}~\cite{goraly2010multi,yamanaka2018colored} the tokens have colours, and we are given an initial and final assignment of coloured tokens to the vertices.
Tokens of the same colour are indistinguishable.  The goal is to move from the initial to the final token arrangement using the fewest swaps.  The original problem is the case where each token has a distinct colour.
Coloured token swapping on graphs is NP-hard for 3 colours~\cite{yamanaka2018colored} but solvable in polynomial time for 2 colours.
In \emph{weighted coloured token swapping} the tokens have colours and each colour has a weight.
Such a weighted colored version has been  studied for string rearrangements under various cost models, which allows swapping non-adjacent elements~\cite{AmirHKLP09}.  

\subsection{Our results} 

A leaf in a tree that already has the correct token is called a \emph{happy leaf}. One feature of all the algorithms for token swapping on trees---both the poly-time algorithms for special cases and the approximation algorithms for the general case---is that they never swap a happy leaf. In 1991 Vaughan~\cite{vaughan1991bounds} conjectured that an optimal swap sequence never needs to swap a token at a happy leaf. We give a 10-vertex counterexample 
to this ``Happy Leaf Conjecture'' 
in Section~\ref{sec:happy-leaf}.

Furthermore, we show in Section~\ref{sec:2-approx} that any algorithm for token swapping on a tree that fixes the happy leaves has approximation factor at least $4/3$, and we show that the two best-known 2-approximation algorithms have approximation factor exactly 2. One insight provided by these results is  that the difficult aspect of token swapping on trees is knowing when and how to swap happy leaves. Our result about the factor 2 lower bound has subsequently been generalized to  any algorithm where no token strays far from its source-destination path~\cite{aichholzer2021hardness}.

Next, we explore \emph{weighted coloured} token swapping. In Section~\ref{sec:weight-colour} we show that weighted coloured token swapping can be solved in polynomial time on paths and stars. In Section~\ref{sec:NP-complete} 
we show that  weighted coloured token swapping is NP-complete for trees. Although this is subsumed by the more recent result that token swapping is NP-complete on trees~\cite{aichholzer2021hardness}, our proof is much shorter, holds even when the trees are restricted to be spiders (i.e., subdivision of stars), and may help to distinguish hard and easy cases of coloured weighted token swapping on trees.

Finally, in an attempt to expand the set of ``easy'' cases, we devised a polynomial time algorithm for token swapping on a \emph{broom}---a star with an attached path---only to discover that this had been done by 
Vaughan~\cite{vaughan1999broom} in 1999, and by Kawahara et al.~\cite{kawahara2016time} in 2016.   
Our simpler proof is in Section~\ref{sec:brooms}.


\subsection{Preliminaries}
\label{sec:preliminaries}

We say that a token is \emph{home} if it is at its destination.  In a tree, \emph{homing} a token means swapping it along the (unique) path from its current position to its destination. 

We defined the token swapping problem as: move token $i$ from its initial vertex to vertex $v_i$, 
with associated permutation $\pi(i) = j$ if token $j$ is initially at $v_i$. An alternative formulation is in terms of an initial and final token assignment.  Suppose $s$ is an initial assignment of tokens to vertices, and $f$ is a final assignment of tokens to vertices.  The goal then is to move each token $i$ from its initial vertex $s(i)$ to its final vertex $f(i)$.  The associated permutation is $\pi(i) = s^{-1}(f(i))$.  (Our first formulation just eases notation by assuming that $f(i) = v_i$.)

A solution to a token swapping problem is a sequence of swaps, $\sigma_1, \sigma_2, \ldots, \sigma_k$.  Our convention is that, starting with the initial token assignment, we perform the swaps starting with $\sigma_1$ and ending with $\sigma_k$ to get the tokens to their final positions. Equivalently, performing the transpositions starting with $\sigma_k$ and ending with $\sigma_1$ generates the associated permutation. 

\section{Background}
\label{sec:background}

This section contains a thorough summary of results on token swapping and related problems.

\subsection{Reconfiguration}
Problems of turning one configuration into another using a limited repertoire of steps have a long history, for example in the study of puzzles and permutation groups~\cite{cayley}.  Recently, the name ``reconfiguration'' has been applied to these problems---see the recent surveys by van den Heuvel~\cite{van-den-heuvel-2013complexity} and Nishimura~\cite{nishimura2018reconfiguration}.
Reconfiguration problems can be formulated in terms of a ``reconfiguration graph'' that has a vertex for each configuration and an edge for each possible reconfiguration step.  As discussed below, when the set of moves forms a permutation group the reconfiguration graph is the Cayley graph.

The general questions that are considered in reconfiguration problems are: can any configuration be reconfigured to any other (connectivity); what is the worst case number of steps required (diameter); and what is the complexity of computing the minimum number of steps required to get from one given configuration to another given configuration (distance). In this paper, we concentrate on distance questions, although we will mention some previous work on deciding reconfigurability and on diameter. We return to more general reconfiguration problems in the final subsection.

\subsection{Permutation groups and Cayley graphs}
\label{sec:perm}
Given a group $(F,*)$, a subset $S$ is a \emph{generator} of $F$, if every element of $F$ can be expressed 
as the product of finitely many elements of $S$ and their inverses. Given a group $F$ and a generator $S$ of $F$, the Cayley graph $\Gamma(F,S)$ has the elements of $F$ as vertices and any two vertices $v,w$ are adjacent, if there exists an element $s\in S$ such that $v * s = w$.
In our context, we are interested in the Cayley graph of the symmetric group $S_n$ that consists of all permutations of the $n$ element set $\{1,\ldots,n\}$.
Given an graph $G = (V,E)$ on $n$ vertices, we define the generating set $S_G$ as  the set of all transpositions corresponding to edges of $E$. The Token Swapping problem corresponds to finding the shortest path in the Cayley graph $\Gamma(S_n,S_G)$ from a given permutation $\pi$ to the identity permutation. Note that this shortest path corresponds to the minimum length generating sequence of $\pi$ by elements of $S_G$. The worst case number of swaps corresponds to the diameter of the Cayley graph. For the special case of token swapping on a path, the Cayley graph is realized geometrically as the graph of a polytope called the \emph{permutohedron} (see~\cite{ziegler2012lectures}).

\subsection{Token swapping on graphs}  
Token swapping on a connected graph of $n$ vertices takes at most $O(n^2)$ swaps---take a rooted spanning tree and, for vertices in leaf-first order, successively \emph{home} the token that goes to that vertex, where \emph{homing} a token means swapping it along the unique path to its final location. This bound is tight for a path with tokens in reverse order. The token swapping problem on graphs (to compute the minimum number of swaps between two given labellings of the graph) is  NP-complete and in fact, APX-complete, as proved by Miltzow et al.~\cite{miltzow2016approximation}, who complemented these hardness results with a polynomial-time 4-approximation algorithm, and an exact exponential time algorithm that is best possible assuming ETH.  These results extend to coloured token swapping.
Bonnet et al.~\cite{bonnet2017complexity} showed that token swapping is W[1]-hard parameterized by number of swaps, but fixed parameter tractable 
for nowhere dense graphs.  This result extends to coloured token swapping and even to a further generalization called ``subset token swapping''. Bonnet et al.~also proved that token swapping is NP-complete on graphs that have tree-width 2 and bounded diameter.

There are many special classes of graphs on which token swapping can be solved via  exact polynomial time algorithms.  These include (in historical order): cliques~\cite{cayley}, paths~\cite{knuth1997art3},   cycles~\cite{jerrum1985complexity}, 
stars~\cite{akers1989group,portier-vaughan1990star,pak1999reduced},
brooms~\cite{vaughan1999broom,kawahara2017time}, squares of paths~\cite{heath2003sorting},
complete bipartite graphs~\cite{yamanaka2015swapping}, and
complete split graphs~\cite{yasui2015swapping}.
See the survey by Kim~\cite{kim2016sorting}.

\subsection{Token swapping on trees}
\label{sec:2-approx-algs}
Various efficient but non-optimal algorithms for token swapping on a tree have been presented in the literature.  
Most of them are 2-approximations---i.e., they use at most twice the optimum number of swaps---although this was not always noted. 
Several of the algorithms are expressed in terms of the paths that tokens should take.
For any token $i$, there is a unique path $p(i)$ from its initial vertex to its final vertex $v_i$.  Let $d(i)$ be the length (the number of edges) of the path $p(i)$, and let $D = \sum_i d(i)$. 

\paragraph*{Happy swap algorithm.} The earliest algorithm we are aware of is due to Akers and Krishnamurthy in 1989~\cite{akers1989group}.  
Their algorithm involves two operations that we will call a ``happy swap'' and a ``shove.'' 
Let $(u,v)$ be an edge with token $i$ on $u$ and token $j$ on $v$.
A \emph{happy swap} exchanges $i$ and $j$ if $p(i)$ includes $v$ and $p(j)$ includes $u$, i.e., the two tokens want to travel in opposite directions across the edge $e$ as the first steps in their paths.  
A \emph{shove} exchanges $i$ and $j$ if $p(i)$ includes $v$ and $j$ is home.
Akers and Krishnamurthy show that: (1) one of these operations can always be applied; and (2) both operations decrease $M = D - (n-c)$ where 
$n$ is the number of vertices and $c$ is the number of cycles in the permutation $\pi$ defined by 
$\pi(i) = j$ if token $i$ is initially at $v_j$. 
Note that if $\pi(i) = i$ (i.e.,~$i$ is home) this forms a trivial cycle which counts in $c$.
Both aspects (1) and (2) of the proof are fairly straightforward.  For (2) they prove that a shove does not change $D$ but decreases $c$, whereas a happy swap decreases $D$ by 2 and changes $c$ by at most 1.
Their proof implies that $M$ is an upper bound on the minimum number of swaps. 
They do not claim that $M$ is at most twice the minimum, but this follows from the easy observation that
$M \le D$ and $D/2$ is a lower bound on the minimum number of swaps, since a single swap decreases $D$ by at most 2. 

Miltzow et al.~\cite{miltzow2016approximation} gave a 4-approximation algorithm for [coloured] token swapping on general graphs.  In case the graph is a tree, their algorithm is the same as the one of Akers and Krishnamurthy and they prove that it is a 2-approximation.

\paragraph*{The Vaughan-Portier algorithm.} 
Independently of the work by Akers and Krishnamurthy, Vaughan and Portier~\cite{vaughan1995algorithm} in 1995 gave an algorithm for token swapping on a tree that uses a number of swaps between $D/2$ and $D$ (in their notation $D$ is called ``PL''). The algorithm involves three operations: {\bf A}, a happy swap; {\bf B}, a version of a happy swap that alters the final token assignment; and {\bf C}, a variant of a shove.  The operations construct the swap sequence by adding swaps at the beginning and the end of the sequence, whereas the other algorithms construct the sequence from the start only.

Operation {\bf B} applies when there is an edge $(u,v)$ and tokens $i$ and $j$ such that 
the destination of $i$ is $u$ and the destination of $j$ is $v$ and $p(i)$ includes $v$ and $p(j)$ includes $u$,
i.e., the two tokens want to travel in opposite directions across the edge $e$ as the last steps in their paths.  
The operation exchanges the final destinations of $i$ and $j$, computes a swap sequence for this subproblem, and then adds the swap of $i$ and $j$ at the end of the sequence. 

Operation {\bf C} applies in the following situation. Suppose there is 
an edge $(u,v)$ with token $i$ on $u$ and token $j$ on $v$, where $p(i)$ 
includes $v$ and token $j$ is home.  Suppose furthermore that there 
is a token $k$ whose destination is $u$ and whose path $p(k)$ includes $v$. 
(Note that this is a more restrictive condition than for a shove.)
The operation exchanges tokens $i$ and $j$ and exchanges the final 
destinations of $j$ and $k$.  Recursively solve this subproblem.  
The swap sequence consists of the swap of $i$ and $j$, followed 
by the sequence computed for the subproblem, followed by the swap of $j$ and $k$.

Vaughan and Portier prove that if operations $\bf A$ and $\bf B$ do not apply, then operation $\bf C$ does, and they prove (this is easy) that each operation decreases the sum of the distances by 2.

\paragraph*{Cycle algorithm.}
The first explicit description of a 2-approximation algorithm for token swapping on trees was given by  
Yamanaka et al.~\cite{yamanaka2015swapping}, who gave an algorithm that sorts the cycles of the permutation one-by-one.  Consider a cycle of tokens $(t_1 t_2 \cdots t_q)$ in the permutation $\pi$.  For $i=1, \ldots, q-1$ their algorithm swaps token $t_i$ along the path from its current vertex to the vertex currently containing token $t_{i+1}$---but stops one short of the destination. Finally, token $t_q$ is swapped from its current vertex to its (original) destination. 

We now outline their proof of correctness and the bound on the number of swaps.
Suppose that token $t_1$ is currently at vertex $x$ and that the first edge it wishes to travel along is $e=(x,y)$.  
Let $j$ be the minimum index, $2 \le j \le q$ such that $t_j$ wishes to travel in the opposite direction along $e$ (and observe that $j$ exists).  
Then the cycle is equivalent to $(t_1 \cdots t_j)$ followed by $(t_j \cdots t_q)$, where the second cycle is empty if $j=q$.  Also, the algorithm performs the same swaps on these two cycles as on the original.
Thus it suffices to prove that their algorithm correctly solves the cycle $(t_1 \cdots t_j)$. This cycle has the special feature that no tokens besides $t_1$ and $t_j$ wish to traverse edge $e$.  
Yamanaka et al.~prove that their algorithm ``almost'' 
achieves the property that just before step $i$ (the step in which $t_i$ moves) tokens $t_1, \ldots, t_{i-1}$ are at their final destinations and all other tokens, including the non-cycle tokens, are at their initial positions. 
``Almost'' means that there is the following exception.  Let $z$ be the vertex containing $t_i$, and let 
$z'$ be the next vertex on the path from $z$ to $a$.  All the tokens on the path from $z'$ to $a$
are one vertex away from their desired positions---they should all be one vertex closer to $z$.
With this exception, the property is obvious for $i=1$ and $i=2$ and can be proved by induction, which implies that the algorithm is correct.
Because tokens are only ``off-by-one'', it can be argued that the number of swaps performed in step $i$ of the algorithm is bounded by the original distance from $t_i$ to its destination.   This implies that the total number of swaps is at most the sum of the distances of labels in the cycle, which gives the factor 2 approximation. 

\paragraph*{Comparisons.}  None of the algorithms will swap a token at a happy leaf, so there is an instance (see Section~\ref{sec:happy-leaf}) where the algorithms are not optimal. The three algorithms differ in how far they allow a token $i$ to stray from its path $p(i)$.
In the Happy Swap algorithm no token leaves the set of vertices consisting of its path together with the vertices at distance 1 from its destination.  
In the Cycle algorithm, no token moves more than distance 1 from its path.  
In the Vaughan-Portier algorithm, a vertex may go further away from its path.

\subsection{Token swapping on paths} 
Token swapping on a path is the classic problem of sorting a list by transposing adjacent pairs.  See Knuth~\cite[Section 5.2.2]{knuth1997art3}.
The minimum number of swaps is the number of inversions in the list. 
Curiously, a swap that decreases the number of inversions need not be a happy swap or a shove (as described above) and, on the other hand, there does not seem to be any measure analogous to the number of inversions that applies to trees more generally, or even to stars.

The diameter of the Cayley graph for token swapping on a path is $\Theta(n^2)$. 
Researchers have also studied the number of permutations with a given number of inversions~\cite{knuth1997art3}, and  
the relationship between the number of inversions and the number of cycles~\cite{edelman1987inversions,balcza1992inversions}.

\subsection{Token swapping on stars} A \emph{star} is a tree with one non-leaf vertex, called the \emph{center vertex}.  
We will need the following known result about token swapping on a star, which expresses the number of swaps as a function of the number of cycles in the permutation $\pi$.  The formula is often written with a delta term whose value depends on whether the center vertex is happy or not, but we will express it more compactly.

\begin{lemma}[{{\cite{akers1989group,portier-vaughan1990star,pak1999reduced}}}]
\label{lemma:stars}  The optimum number of swaps 
to sort an initial placement of tokens on a star is 
$n_U + \ell$, where $n_U$ is the number of unhappy leaves and 
$\ell$ is the number of cycles in the permutation that have length at least 2 and do not involve the center vertex.
\end{lemma}
\begin{proof}[sketch]
Consider a cycle $C$ of length at least 2 in the permutation of tokens and consider the corresponding vertices of the star. If the center vertex is not in $C$ then the number of swaps to sort $C$ is its number of leaves plus one.  If the center vertex is in $C$ then the number of swaps is the number of leaves in $C$.
Because the cycles are independent, we can sum over all non-trivial cycles, which yields the stated formula.
\end{proof}

It follows that the diameter of the Cayley graph for a star is ${\frac 3 2}n + O(1)$, which arises when all cycles have length 2. 
Further properties of Cayley graphs of stars were explored by Qiu et al.~\cite{qiu1994star-pancake}.
Portier and Vaughan~\cite{portier-vaughan1990star} analyzed the number of vertices of the Cayley graph at each distance from the distinguished ``sorted'' source vertex (see also~\cite{wang2006distance}).  Pak~\cite{pak1999reduced} gave a formula for the number of 
shortest paths between two vertices of the Cayley graph.

\subsection{Transposition trees and interconnection networks}
The network community's interest in token swapping on trees (``transposition trees'') stems from the use of the corresponding Cayley graphs as interconnection networks, an idea first explored by Akers and Krishnamurthy \cite{akers1989group}.  Cayley graphs of transposition trees have the following desirable properties: they are large graphs ($n!$ vertices) that are vertex symmetric, with small degree ($n-1$), large connectivity (the same as the degree), and small diameter.  In particular, the diameter is ${\frac 3 2}n + O(1)$ when the tree is a star.  
The commonly used hypercube has $2^n$ vertices and diameter $n$, so the diameter is logarithmic in the size.  By contrast, 
the Cayley graph of a star has sublogarithmic diameter.

Akers and Krishnamurthy proved a bound on the diameter of the Cayley graph of a transposition tree, specifically, the maximum over all permutations of 
the bound $D - (n-c)$ which was discussed above.  This bound cannot be computed efficiently since it involves the maximum over $n!$ permutations.
Vaughan~\cite{vaughan1991bounds} also gave upper and lower bounds on the diameter of the Cayley graph, though neither easy to state nor to prove.

Follow-up papers by Ganesan~\cite{ganesan2012efficient}, Chitturi~\cite{chitturi2013upper} and Kraft~\cite{kraft2015diameters} have lowered the diameter bound and/or the time required to compute the bound.
To give a flavour of the results, we mention a polynomial-time computable upper bound, $\gamma$, due to Chitturi~\cite{chitturi2013upper} that is defined recursively as follows: if the tree is a star, use the known diameter bound; otherwise
choose a vertex $v$ that maximizes the sum of the distances to the other vertices, increase $\gamma$ by the maximum distance from $v$ to another vertex
 and recurse on the smaller tree formed by removing the leaf $v$. 


\subsection{Happy leaves} 
As mentioned in the introduction above, Vaughan~\cite{vaughan1991bounds} conjectured that a happy leaf in a tree need not be swapped in an optimal swap sequence.
In fact she made a stronger conjecture~\cite[Conjecture 1]{vaughan1991bounds} that if a tree has an edge $(a,b)$ such that no token wishes to cross $(a,b)$ (i.e., no path from a token to its destination includes edge $(a,b)$) then there is an optimal swap sequence in which no token swaps across $(a,b)$.  The Happy Leaf Conjecture is the special case where $b$ is a leaf.

Smith~\cite[Theorem 9]{smith1999factoring} claimed something stronger than the happy leaf conjecture: that no optimal swap sequence would ever swap a happy leaf.  But later he found an error in the proof~\cite{smith2011corrigendum}, and gave an example of a small tree where there is an optimal swap sequence that performs a swap on a happy leaf.  In his example, there is also an optimal swap sequence that does not swap the happy leaf so he did not disprove the happy leaf conjecture.

\subsection{Coloured token swapping}
\label{sec:background-colouredTS}
Many natural reconfiguration problems involve ``coloured'' elements, where two elements of the same colour are indistinguishable.  Token swapping for coloured tokens was considered by 
Yamanaka et al.~\cite{yamanaka2015colored} (journal version~\cite{yamanaka2018colored}).
They proved that the coloured token swapping problem is NP-complete for $c \ge 3$
colours even for planar bipartite graphs of max degree 3, but for $c=2$ the problem is solvable in polynomial time, and in linear time for trees.
On complete graphs, coloured token swapping is NP-complete~\cite{AmirP15,bonnet2017complexity} but 
fixed parameter tractable in the number of colours~\cite{yamanaka2018colored}. Coloured token swapping can be solved in polynomial time on a path~\cite{AmirP15,bonnet2017complexity} and on a star~\cite{bonnet2017complexity}.

\subsection{More general token and pebble games}
There are many problems similar to token swapping, some of which were studied long ago.  We mention some results but do not attempt a complete survey.

The classic 15-puzzle (see~\cite{Demaine2009games,van-den-heuvel-2013complexity} for surveys) is a version of token swapping on a $4 \times 4$ grid graph where only one special token (the ``hole'') can be swapped.  As a consequence, only half the token configurations (the alternating group) can be reached.  
Generalizing beyond the $4 \times 4$ grid to general graphs, Wilson~\cite{wilson1974graph} in 1974 gave a complete characterization of which token configurations on which graphs can be reached via this more limited set of moves.  Minimizing the number of moves is NP-complete~\cite{goldreich2011finding} even for grid graphs~\cite{ratner1990n2}. 
Recently, a version with coloured tokens has also been considered~\cite{yamanaka2017sequentialcolored}.  

Several papers explore generalizations where there is more than one ``hole''.  As in the 15-puzzle, a swap must involve a hole.   The tokens may be labelled or unlabelled.
For labelled tokens, Kornhauser et al.~\cite{kornhauser1984coordinating} in 1984  
gave a polynomial time algorithm to decide if reconfiguration between two label placements is possible, and proved a tight bound of $O(n^3)$ on the diameter of the associated Cayley graph.

For unlabelled tokens and holes Auletta et al.~\cite{auletta1999linear} gave a linear time algorithm to decide reconfigurability on any tree, and this was generalized to linear time for any graph (even for coloured tokens) by Goraly and Hassin~\cite{goraly2010multi}.
C\u{a}linescu et al.~\cite{calinescu2008reconfigurations} gave a polynomial time algorithm to minimize the number of swaps for any graph, 
but showed that the problem becomes APX-hard if the objective is to minimize the number of times a token moves along a path of holes.
Fabila-Monroy et al.~\cite{fabila2012token} showed that the diameter of the Cayley graph (which they call the ``token graph'') is at most $k$ times the diameter of the original graph.

Papadimitriou et al.~\cite{papadimitriou1994motion} considered a ``motion planning'' version 
where all the tokens are unlabelled ``obstacles'' except for one labelled ``robot'' token which must move from a start vertex to a destination vertex; as before, a swap must involve a hole.  
They showed that minimizing the number of swaps is NP-complete for planar graphs but solvable in polynomial time for trees.  The run time for trees was improved in~\cite{auletta2001optimal}.  

A unifying framework is to replace holes by labelled or coloured tokens that are ``privileged,'' and to require that every swap must involve a privileged token.

In another variant, token movement must be carried out by a single robot walking along the graph edges and carrying at most one token at a time.
Graf~\cite{graf2017sort} includes a good summary. 

Rather than token swapping across an edge, an 
alternative motion is rotation of tokens around a simple cycle in the graph. This is of interest in the robotics community since it models movement of robots with the restriction that no two robots can travel along the same edge at the same time. 
When all cycles in a graph may be used, there
are polynomial time algorithms to decide if reconfiguration is possible~\cite{yu2015pebble,foerster2017multi}.
See also~\cite{yu2016intractability} for hardness of optimization and~\cite{yu-lavalle2016optimal} for practical approaches.
If rotation is only allowed around the cycles of a cycle basis (e.g.,~the faces of a planar graph)
Scherphuis~\cite{scherphuis-rotation} provided a characterization (similar to Wilson's for the 15-puzzle generalization) of which graph/cycle-basis/token- placement combinations permit reconfiguration (see also Yang~\cite{yang2011sliding}, who showed that Wilson's result reduces to this result).

Although the term ``pebbling'' is sometimes used for token swapping, the more standard distinction is that \emph{pebbling games} are played on directed acyclic graphs and involve rules to add and remove pebbles, modelling memory usages or dependencies in computation.  Nordstr\"om~\cite{nordstrom2013pebble} wrote a thorough survey of pebbling games and their relevance to complexity theory.

\section{Counterexample to the Happy Leaf Conjecture}
\label{sec:happy-leaf}

In this section we disprove the Happy Leaf Conjecture by giving a tree with initial and final token placements such that any optimal swap sequence must swap a token on a happy leaf (recall that a happy leaf is one that has the correct token).  Our counterexample has $n=10$ vertices and is shown in  
Figure~\ref{fig:happy-swap-counter-ex}(a).
This is the smallest possible counterexample---we have verified by computer search that all trees on less than 10 vertices satisfy the Happy Leaf Conjecture.
Our counterexample can easily be generalized to larger $n$, and we give more counterexamples in the next section.
Our tree consists of a path $v_1, \ldots, v_9$ and one extra leaf $v_{10}$ joined by an edge to vertex $v_3$.  
The initial token placement has token 10 at $v_{10}$ (so $v_{10}$ is a happy leaf) and tokens $9, 8, \ldots, 1$ in that order along the path $v_1, v_2, \ldots, v_9$.

\begin{figure}[t]
    \centering
    \includegraphics[width=.9\textwidth]{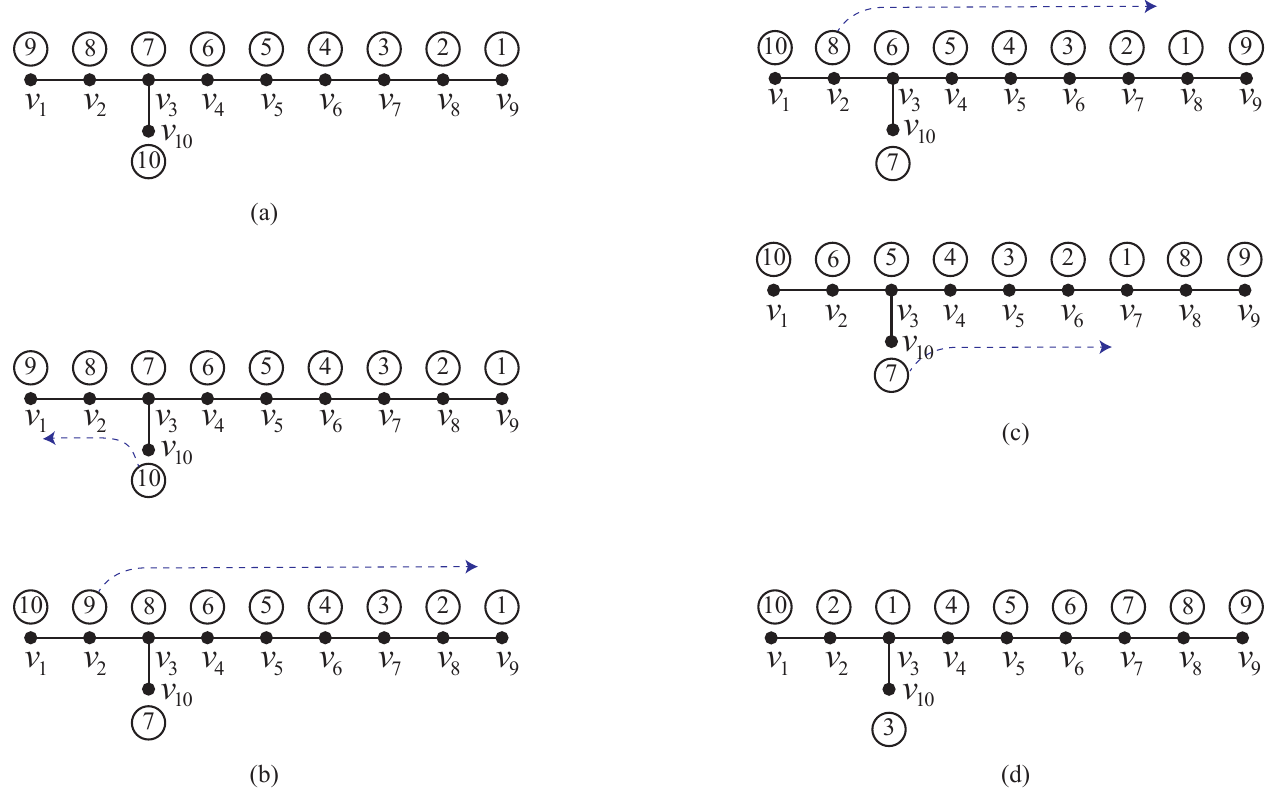}
    \caption{A counter-example to the happy leaf conjecture where an optimum swap sequence involves moving the happy token 10. (a) The initial tokens (in circles). (b) Three swaps move token 10 to $v_1$. (Dashed blue arrows show the upcoming moves.) (c) The result of homing tokens 9 and 8. (d) The result of homing tokens 9 through 4.  Four additional swaps will sort the tokens. }
    \label{fig:happy-swap-counter-ex}
\end{figure}

If token 10 does not leave vertex $v_{10}$ (i.e., we fix the happy leaf), then we must reverse the order of the tokens on a path of length 9, which takes $\binom{9}{2} =36$ swaps.
However, as we now show, there is a swap sequence of length 34.

Initially, we perform  3 swaps to move token 10 to $v_1$ giving the configuration shown in Figure~\ref{fig:happy-swap-counter-ex}(b). 
Next, we ignore leaf $v_1$ and perform the sequence of swaps that homes tokens $9,8, \ldots, 4$.  The result of homing tokens 9 and 8 is shown in Figure~\ref{fig:happy-swap-counter-ex}(c), and the result of homing all of them is shown in Figure~\ref{fig:happy-swap-counter-ex}(d).
It is easy to verify that this takes 7 swaps for token 9, 6 for token 8, $\ldots$ , 2 for token 4, which adds up to $7 + 6 + \cdots + 2 = 27$ swaps.

Finally, we perform the following swaps to complete the sort: home token 10 in 3 swaps, then home token 1 in 1 swap.
In total, this uses $3 + 27 + 4 = 34$ swaps.

The idea of why this saves swaps is as follows.  To reverse the order of edges on a path, every token must swap with every other token.  By contrast, the above approach saves swaps whenever two tokens occupy vertices $v_2$ and $v_{10}$.  For example, tokens 8 and 7 never swap with each other, nor do 7 and 6, etc. 
We need $n \ge 10$ so that this saving exceeds the cost of the initial set-up and final clean-up.

\section{Lower Bounds on Approximation}
\label{sec:2-approx}
 
In this section we prove that if an algorithm for token swapping on a tree never swaps a token on a happy leaf, then it has worst case approximation factor at least $\frac{4}{3}$.
In Section~\ref{sec:2-approx-algs} we described three 2-approximation algorithms for token swapping on a tree: the Happy Swap algorithm, the Vaughan-Portier algorithm, and the Cycle algorithm. All three algorithms fix the happy leaves, so the lower bound applies.
For the Happy Swap algorithm and the Cycle algorithm we improve the approximation lower bound to 2.  To do this we use further properties of those algorithms, in particular, constraints on how far a token may stray from the path between its initial and final positions.  The question of whether the Vaughan-Portier algorithm is also limited by factor 2 was recently answered~\cite{aichholzer2021hardness}.

\begin{theorem}
\label{thm:4-3-approx}
Any algorithm that does not move tokens at happy leaves has an approximation factor of at least $\frac{4}{3}$.
\end{theorem}
\begin{proof}
Define a tree $T_k$ to have a path of $2k+1$ vertices, $p_k, p_{k-1}, \ldots, p_1, p_0, p'_1, \ldots, p'_k$ together with a set $L$ of $k$ leaves adjacent to the center vertex, $c = p_{0}$, of the path.  
The tokens at $c$ and the vertices of $L$ are already home, 
i.e., they have the same initial and final positions. 
The tokens on the path should be reversed, i.e., the token initially at vertex $p_i$ of the path has final position $p'_i$ and vice versa.
See Figure~\ref{fig:4-3-approx-example} for an illustration.

\begin{figure}[ht]
    \centering
    \includegraphics[width=.9\textwidth]{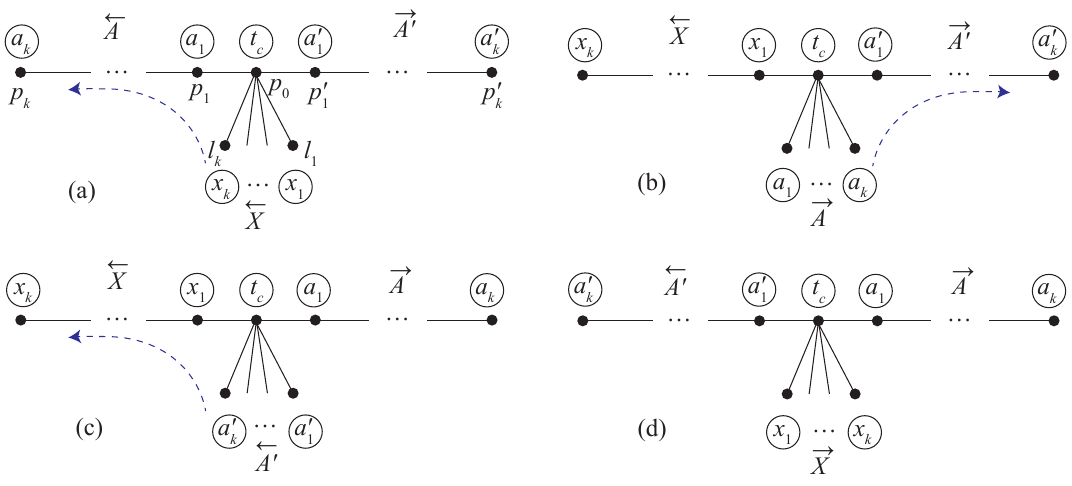}
    \caption{Illustration for Theorem~\ref{thm:4-3-approx}. 
    (a) The tree $T_k$ with initial token placement. 
    Notation $\protect\overleftarrow{A}$ indicates the order of the labels (in this case right to left). 
    The blue dashed arrow indicates the first move on tokens $X$ in step (1) of swap sequence $\cal S$.  
    (b) After step (1) $A$ and $X$ are exchanged. Note that $A$ is reversed.  (c) After step (2).  (d) After step (3).  It remains to reverse the tokens $X$. 
    }
    \label{fig:4-3-approx-example}
\end{figure}

Any algorithm that fixes happy leaves must reverse the path which takes $\binom{2k+1}{2} = 2k^2 + k$ swaps.

We now describe a better swap sequence $\cal S$.  Let $t_c$ denote the token at vertex $c$, 
let $A = \{a_1, a_2, \ldots, a_k\}$ denote the tokens initially at $p_1, \ldots, p_k$, let $X = \{x_1, \ldots , x_k\}$ denote the tokens initially at vertices $\ell_1, \ldots, \ell_k$ of $L$, and let $A' = \{a'_1, \ldots, a'_k\}$ denote the tokens initially at $p'_1, \ldots, p'_k$.
The plan is to perform the following 4 steps: (1) exchange the tokens $A$ with the tokens $X$; (2) exchange $A$ with $A'$; (3) exchange $A'$ with $X$; and (4) adjust the order of tokens $X$. 

For clarity of explanation, we will implement each step in a way that is suboptimal but efficient enough for our purposes. 
For step (1), we implement the exchange of the two token sets as follows:
first move token $t_c$ to $p_k$ using $k$ swaps, then for $i=k, \ldots, 1$ move token $x_i$ to vertex $p_{i-1}$ using $i$ swaps, and finally, move token $t_c$ from $p_k$ back to $c$ using $k$ swaps. This places tokens $X$ on the first half of the path in the correct order ($x_i$ on $p_i$), and tokens $A$ on vertices $L$ in reversed order ($a_i$ on $\ell_{k+1-i}$). The total number of swaps is $2k + \binom{k+1}{2} = \frac{1}{2}k^2 + \frac{5}{2}k$. 

Step (2) acts similarly on $A$ and $A'$, after which tokens $A$ are home, and tokens $A'$ are on vertices $L$ in the correct order ($a'_i$ on $\ell_i$).
Finally, after step (3), tokens of $A$ and $A'$ are home, and tokens $X$ are on $L$ but in reverse order.  
For step (4) we must sort the tokens on $X$ which means solving a tree swapping problem on a star of $k$ leaves, where the permutation consists of $k/2$ cycles each of length 2 ($x_i$ must be exchanged with $x_{k+1-i}$).  Each cycle takes 3 swaps for a total of $\frac{3}{2}k$ swaps.

The total number of swaps used is $\frac{3}{2}k^2 + 9k$.
This is an upper bound on the optimum number of swaps.
The ratio of number of swaps used by the algorithm to optimum number of swaps is at least $(2k^2 + k)/(\frac{3}{2}k^2 + 9k) = \frac{4}{3} - o(1)$.
\end{proof}

\begin{theorem}\label{thm:Tight2}
    The Happy Swap and the Cycle algorithms  
    do not have an approximation factor less than $2$.
\end{theorem}

\begin{proof}
Our example generalizes the one in the preceding proof to have $b$ paths of $k$ vertices rather than two paths. 
For every $k$ and every odd $b$ we define a tree $T_{k,b}$ together with initial and final token placements.
The tree $T_{k,b}$ consists of $b$ paths of $k$ vertices all joined to a center vertex $c$. 
Additionally, there is a set $L$ of $k$ leaves adjacent to $c$. 
See Figure~\ref{fig:T43} for an illustration.

\begin{figure}[ht]
    \centering
    \includegraphics{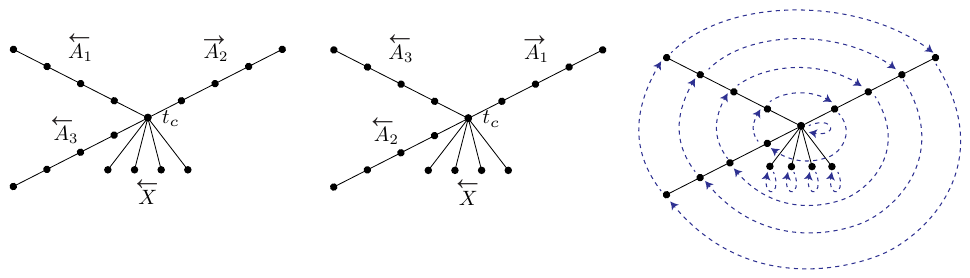}
    \caption{Left: The tree $T_{4,3}$ with initial token placement. 
    Notation $\protect\overleftarrow{A}$ indicates the order of the labels.
    Middle: The tree with its target token placement. Right: The tree with arrows that indicate for each token the initial and target placement.
    }
    \label{fig:T43}
\end{figure}

We denote the paths of $k$ vertices by $P_1,\ldots,P_b$ respectively. 
Let $A_i$ be the tokens initially on path $P_i$, let $X$ be the tokens initially on $L$, and let $t_c$ be the token on $c$. 

In the final token placement the tokens $A_i$ are on path $P_{i+1}$, indices modulo $b$. 
Each token in $A_i$ should have the same distance
to $c$ in the initial and final token placements.
The tokens $X$ and $t_c$ should stay where they are.

\paragraph{Idea of the proof.}
Any solution for token swapping on this input has to accomplish two tasks.
Each set of tokens $A_i$ has to be moved from its initial path $P_i$ to its final path $P_{i+1}$.
This takes roughly $k^2$ swaps, as $k$ tokens have 
to move $k$ steps. But each those swaps could accomplish something else 
for another token, so only half of those swaps should be charged
to this task.
The other task is to reverse the order of each set $A_i$. 
Now, we know that this takes $\binom k 2 \approx k^2/2$ swaps on 
a path. But if we can use the happy leaves, this can 
be done with $O(k)$ swaps and thus is asymptotically free.

More precisely, we will first give a clever solution using happy leaves to show that the optimum number of swaps, $S_{\it OPT}$, satisfies 
$S_{\it OPT} \le (b+1)(\binom {k+1}{2} + 2k)$.
Then we will argue that the number of swaps, $S_A$, used by either algorithm is at least $k^2 b$. 

To complete the proof, observe that
for $k=b$, $S_{\it OPT} \le \frac{1}{2}k^3 + o(k^3)$ and $S_A \ge k^3$.
Then for every $\varepsilon > 0$ there is a $k$
such that
$S_A / S_{\it OPT} > 2-\varepsilon$.
This shows that the algorithms do not have an approximation factor better than 2.

We now fill in the details of the clever solution, and the analyses of the algorithms. 

\paragraph{A Clever Solution.}  
We describe a procedure to swap every token to its target position using only $(b+1)(\binom {k+1}{2} + 2k)$ swaps. 
See Figure~\ref{fig:CleverSwapping} for an illustration.
\begin{figure}
    \centering
    \includegraphics{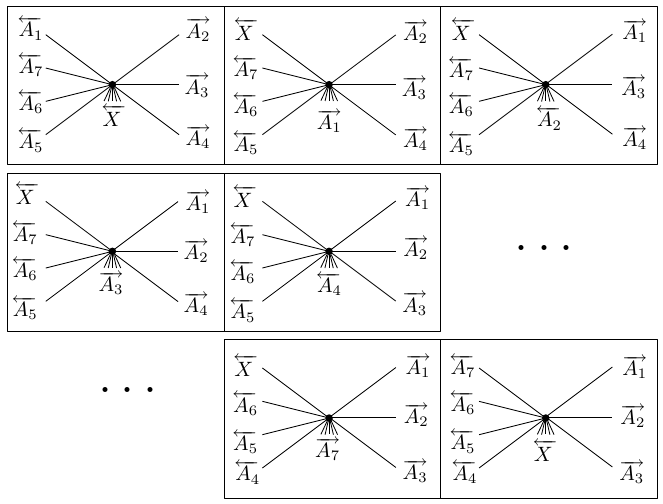}
    \caption{A clever way of swapping all the tokens.}
    \label{fig:CleverSwapping}
\end{figure}
As a first step, exchange the tokens $A_1$ with the tokens $X$.  The exchange is implemented as described in the proof of 
Theorem~\ref{thm:4-3-approx} and takes $\binom {k+1}{2} + 2k$ swaps.  
In the second step, exchange the tokens $A_1$ with the tokens $A_2$.  This places the tokens of $A_1$ on path $P_2$ in the correct order, and 
again takes $\binom {k+1}{2} + 2k$ swaps.
In general, the $i{\rm th}$ step exchanges tokens $A_{i-1}$ (currently on the vertex set $L$) with the tokens $A_i$, moving $A_{i-1}$ onto the path $P_i$ in the correct order.
In the $(b+1){\rm st}$ step, tokens $A_b$ are exchanged with tokens $X$, moving $A_b$ onto the path $P_1$.
If $b$ is odd, the tokens $X$ will end up in the correct order on $L$, and the total number of swaps is $(b+1)(\binom {k+1}{2} + 2k)$.

\paragraph{Behavior of the Happy Swap algorithm.}
We will prove that the Happy Swap algorithm uses at least $k^2 b$ swaps.
The tokens of $A_i$ must move from path $P_i$ to path $P_{i+1}$.
Suppose $P_i$ has vertices $p_1, \ldots, p_k$ where $p_1$ is adjacent to $c$.  Suppose $P_{i+1}$ has vertices $p'_1, \ldots, p'_k$ where $p'_1$ is adjacent to $c$. Let $a_j$ be the token whose initial vertex is $p_j$ and whose final vertex is $p'_j$.
The path from initial to final vertex is $p_j, p_{j-1}, \ldots, p_1, c, p'_1, \ldots, p'_j$, and its length is $2j$. We note the following properties:

\smallskip\noindent(*)
Until $a_j$ reaches $p'_j$, happy swaps and shoves can only move $a_j$ along its path.
Once $a_j$ reaches $p'_j$, a shove can move $a_j$ to
$p'_{j-1}$ or to $p'_{j+1}$, after which it must move back to $p'_j$. 

\smallskip\noindent
This implies, in particular, that $a_j$ cannot leave $P_i \cup P_{i+1}$.  

On the long path $P_i \cup P_{i+1}$ the tokens $a_j$ must reverse their order.  Because none of these tokens leave the long path, the algorithm must perform a swap for each of the $\binom{k}{2}$ inversions.  The algorithm must also perform some swaps to move tokens from $P_i$ to $P_{i+1}$.
We need a way to count these separately because a single swap can make progress towards both goals.
We will use a charging argument.

Each swap involves two tokens, and \emph{pays} 
$\frac{1}{2}$ to each of those two tokens. The number of swaps is the sum, over all tokens, of the payments to that token.  
Now we will charge the work performed to the tokens.  
Each of the $2j$ edges 
of $a_j$'s path charges $\frac{1}{2}$ to token $a_j$.  Each inversion $a_j, a_h$, $h > j$ charges 1 to $a_j$.  
The amount charged to $a_j$ is $j + (k-j) = k$.

We will prove below that for every token $a_j$, the payment is at least the charge.  
This implies that the total number of swaps is at least $bk^2$ since there are $bk$ tokens.

\begin{claim}
The payments to $a_j$ are at least the charges to $a_j$.
\end{claim}
\begin{proof}
The first time $a_j$ reaches its destination, it has been paid $\frac{1}{2}$ for each edge on its path.  This balances the charges for the path.

It remains to account for the inversions charged to $a_j$.
Consider an inversion $a_j, a_h$, $j<h$.  Some swap must change the order of these tokens in the path $P_i \cup P_{i+1}$, and, in this swap, $a_j$ moves in a direction opposite to the direction of its path from initial to final vertex.  Thus, by property (*), this swap can only happen after $a_j$ has reached its destination. 
In particular, there are two possibilities for this swap: (1) token $a_j$ is home at $p'_j$ and token $a_h$ is at $p'_{j-1}$ and a shove exchanges them; or (2) token $a_j$ is at $p'_{j+1}$ and token $a_h$ is at $p'_j$ and a happy swap exchanges them.  In case (1), there must be a subsequent swap where $a_j$ returns to its home.  In case (2) there must have been a previous swap where $a_j$ was shoved from its home to $p'_{j+1}$.  In both cases the extra swap moves $a_j$ in the direction of its path, and cannot fix an inversion.  Thus in both cases, there are two swaps to pay for the inversion.     
\end{proof}

\paragraph{Behavior of the Cycle algorithm.}
We will prove that the Cycle algorithm uses at least $k^2 b$ swaps.
The Cycle algorithm solves the cycles independently.  Each non-trivial cycle has $b$ tokens, one on each path.  
Consider the action of the algorithm on the cycle of tokens at distance
$i$ from $c$. Let $t_j$ be the token on path $P_j$ at distance $i$ from $c$. 
Suppose without loss of generality that the algorithm moves $t_1$ first.  
It swaps $t_1$ along $P_1$ to the center vertex $c$, which takes $i$ swaps, 
and then along $P_2$ from the center vertex to the vertex just before 
the one containing $t_2$, which takes $i-1$ swaps.  The positions of 
the other tokens on the cycle have not changed.  
Thus, each $t_j$ takes $2i -1$ swaps, and finally, $t_b$ takes $2i$ swaps.  The total over all cycles 
is $\sum_{i=1}^k (b(2i-1) + 1) = 2b \binom{k+1}{2} -bk + k  = bk^2 + k$.
\end{proof}


\section{Weighted Coloured Token Swapping is NP-complete}
\label{sec:NP-complete}

In the \emph{coloured} token swapping problem, we have coloured tokens, one on each vertex of the graph, and each vertex has a colour.  
The goal is to perform a minimum number of swaps to get each token to a vertex that matches its colour. 
We assume that the number of tokens of each colour is equal to the number of vertices of that colour. 
The standard token swapping problem is the special case when all colours are distinct.

For general graphs the coloured token swapping problem can be solved in polynomial time for 2 colours, but becomes  NP-complete for $k \ge 3$ colours~\cite{yamanaka2018colored}.
See Section~\ref{sec:background-colouredTS} for further background.

In the \emph{weighted} coloured token swapping problem,
each colour $c$ has a weight $w(c)$, and the cost of
swapping tokens of colours $c$ and $c'$ is 
$w(c) + w(c')$.  The goal is to reach the target configuration with minimum total cost.

In this section we prove that weighted coloured token swapping is NP-complete.  This result is subsumed by the very recent and difficult result that token swapping (with no colours or weights) is NP-complete~\cite{aichholzer2021hardness}. 
Our proof is much shorter, holds even when the trees are restricted to be spiders (i.e., subdivision of stars),  reduces from a different problem, and may be of value in trying to distinguish between hard and easy cases of coloured weighted token swapping.  There is currently no known class of trees where token swapping is in P but weighted coloured token swapping is NP-complete.

Our reduction is from vertex cover in graphs with maximum degree three~\cite{Garey:1979:CIG:578533}.  We construct a long path with some `green' 
tokens initially at the right end of the path. The final configuration has green tokens at the left end of the path.  We can save the cost of moving all those green tokens the whole length of the path by dislodging some happy green tokens from a subtree that dangles off the path; we construct this dangling subtree from the vertex cover instance in such a way that there is a cost savings if and only if there is a small vertex cover.   

\begin{theorem}
Weighted coloured token swapping on spiders is NP-complete.
\end{theorem}
\begin{proof}
We will reduce the NP-complete problem vertex cover in a graph  with maximum degree three~\cite{Garey:1979:CIG:578533} to weighted coloured token swapping on spiders. Let the input to vertex cover be a graph $G$ with $n$ vertices, $m$ edges, and maximum degree three, together with a number $q$.  The problem is whether $G$ has a vertex cover of size at most $q$.

We first construct a spider $T$  with root $r$ and $2n$  children of the root.  Here $n$ of these children  correspond to the vertices of $G$ which we refer to as \emph{$V$-vertices} and the remaining  $n$ children, $o_1,\ldots,o_n$, are called  \emph{helper vertices},  e.g., see Figures~\ref{fig:hard}(a)-(b). For each vertex $v$ of $G$ we now add a path of  $4\deg (v)$ vertices connected to the $V$-vertex for $v$.  Specifically,
for each edge $(x,y)$ of $G$, we create eight  
vertices in $T$  and we add four of them  to the path of 
the $V$-vertex corresponding to $x$ and the other four to the path of the $V$-vertex corresponding to $y$. 
We will refer to these new  vertices as the \emph{$E$-vertices} of $T$. 

We place tokens on $T$  in the initial and final configurations (which are identical) as follows. 
For each vertex $w$ of $G$, we place a token of color $w$
 and weight $n^5$ on the corresponding $V$-vertex. In the schematic representation (see Figure~\ref{fig:hard}(c)) 
the tokens of $V$-vertices are shown in gray,  
but note that they are in fact a collection of $n$ different  colours. We place an orange token of weight 1 on each helper vertex. On the root of $T$ we place a blue token with weight 1.  On each of the eight 
$E$-vertices $(x,y)$ we place a token of colour $xy$ and weight 1. In the schematic representation (see Figure~\ref{fig:hard}(c)) 
the tokens of $E$-vertices are shown in green,  
but note that they are in fact a collection of $m$ different  colours, with eight tokens of each color. 
We will refer to these colours as \emph{edge-colours}.

We then construct a path $P$ of length  $(L_r+4m+(L_r-1)+4m)$, where 
$L_r = n^7$.  We identify the root of $T$ with the $(L_r+4m)^{\rm th}$ vertex of the path $P$. Let $T'$ be the resulting  spider. 

The initial configuration of tokens on $P$ is as follows (see Figure~\ref{fig:hard}(d)): 
\begin{enumerate}
    \item[-] the first $L_r$ vertices 
    starting from the left have red tokens of weight 1,  
    \item[-] the next $4m+(L_r-1)$ vertices have blue tokens of weight 1,
    \item[-] the next $4m$ vertices have four tokens of each of the $m$ edge-colours, each of weight 1.
\end{enumerate}

The final configuration of tokens on path $P$ is shown in Figure~\ref{fig:hard}(g) 
(ignore parts (e) and (f) for now).  The red tokens are now at the right end of the path and the edge-coloured tokens are at the left end.

\begin{figure}[pt]
\centering
\includegraphics[width=.87\textwidth]{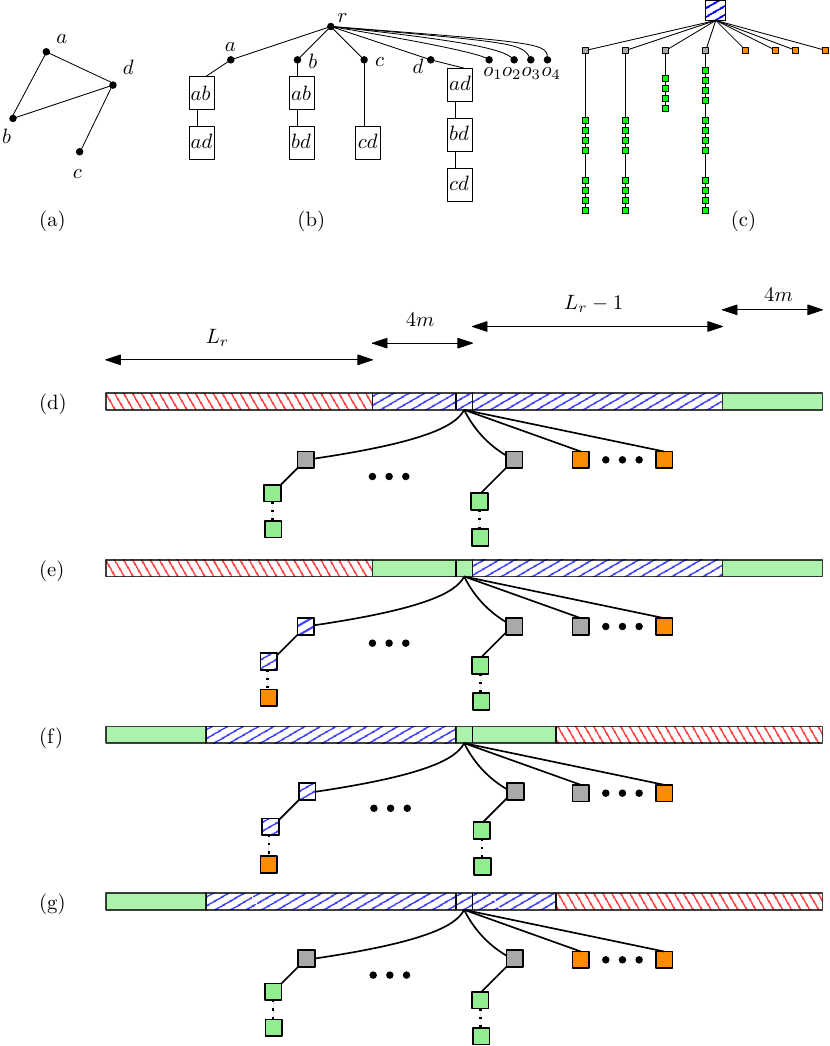}
\caption{Illustration for the NP-hardness proof: (a) input graph for vertex cover; (b) schematic of  spider  $T$; (c) details and initial configuration of coloured tokens on $T$;  
(d) spider  $T'$ with initial configuration of coloured tokens;
(e)--(f) swapping from initial to final configuration of $T'$; (g) $T'$ with final configuration of coloured tokens.  Note that, rather than drawing tokens alongside the vertices as in the other figures, we are simply colouring the vertices by their initial/final/current token colours.  
}
\label{fig:hard}
\end{figure}

Let $\alpha = 2q(1+n^5) + 24q$, $\beta = 16m^2 + 92m$, $\gamma = 2L_r(8m + (L_r - 1))$. In the following we show that $G$ has a vertex cover of size $q$ if and only if 
the weighted coloured token swapping problem 
on $T'$ has a solution of cost at most  $2\alpha + 2\beta + \gamma$.

\paragraph{From Vertex cover to Weighted coloured Token Swapping:}
Assume that $G$ has a vertex cover $\mathcal{C}$ of size $q$. We reach the target configuration as follows: 
\begin{description}
    \item[Move green tokens from $T$ to $P$, Fig.~\ref{fig:hard}(d)--(e):] 
    There are two substeps:
    \begin{enumerate}
    \item For each $V$-vertex $v$ of $\mathcal{C}$ move its gray token  
    to a helper vertex.  
    This requires two swaps, each with a cost of $(1+n^5)$. After the gray token at $v$ moves, an orange token ends up at the root of $T$.  We then move this orange token to the leaf descending from $v$.
    (Note that this returns the original blue token to the root of $T$.)
    Since a root-to-leaf path in $T$ contains at most 12 vertices, this takes
    at most 12 swaps, each with a cost of 2.
    Since there are $q$ gray tokens, the total cost of the first step is $2q(1+n^5) + 24q = \alpha$.

    \item Let $\mathcal{D}$ be the edge-coloured (green) tokens on the paths descended from the $V$-vertices of $\mathcal{C}$.  
    In the second step, we move $4m$
    tokens of $\mathcal{D}$, 
     four of each colour, into
    the path $P$, 
    swapping them with $4m$ blue tokens from the path.
    We move the tokens one after the other in the order they appear at the end of the path $P$ in the final configuration. 
    Hence the first edge-coloured token needs to travel at most $12$ steps to the root of $T$ plus $4m-1$ steps along $P$, i.e., $(4m+11)$ steps, the next one needs to travel at most $12$ steps to the root plus $4m -2$ steps along $P$, i.e.,  $(4m+10)$ steps and so on. Hence the total cost of the second step is at most $ 2((4m+11)+(4m+10) + \ldots + 12) = (4m+11)(4m+12)-132 = 16m^2 + 92m = \beta$.
    \end{enumerate}
    The total cost is $\alpha + \beta$.

    \item[Move red tokens along $P$, Fig.~\ref{fig:hard}(e)--(f):] Move the red tokens all the way to the right end of  $P$. Moving each token takes $8m+(L_r-1)$ swaps. Hence for $L_r$ tokens, this costs $\gamma = 2L_r(8m+(L_r-1))$.
    Note that this has the effect of moving the green tokens that were initially at the far right to the root of $T$.

    \item[Move green tokens from $P$ to $T$, Fig.~\ref{fig:hard}(f)--(g):]
    This step is the reverse of the first one.  We swap the edge-coloured (green) tokens on $P$ with the blue tokens that are inside $T$, and then move the orange and gray tokens back to their initial ($=$ final) positions. The cost is $\alpha + \beta$. 
\end{description}

The total cost to reach the final configuration is $2\alpha + 2\beta + \gamma$.

\paragraph{From Weighted coloured Token Swapping to Vertex cover:} Assume that the weighted coloured token swapping problem on $T'$ has a solution $\mathcal{S}$ of cost at most $2\alpha + 2\beta + \gamma$. We now show how to find a vertex cover of size at most $q$.

Observe first that we need to pay at  least a cost of $\gamma$ for moving the red tokens from the left to the right end of $P$. This leaves a remaining cost of at most 
 $2\alpha + 2\beta$, which is what we analyze in the rest of the proof.
 
If $\mathcal{S}$ does not move the tokens of $V$- and $E$-vertices (the non-root vertices of $T$), then we 
need to move the edge-coloured (green) tokens from the right end of $P$ to the left end of $P$. Moving even one green token from the right end of $P$ to one of the leftmost $4m$ positions of $P$ costs at least $L_r = n^7$. One can verify that $n^7 > 2(16(3n)^2+92(3n))+4n(1+n^5)+48n$ (for every $n\ge 5$), so for $m \le 3n$ (a loose upper bound) and $\beta = 16m^2+92m$, we obtain 
$L_r  > 2\alpha + 2\beta $. Consequently, all the green tokens that are   at the left end  of $P$ in the final configuration must come from the $E$-vertices. Since the gray tokens are `blocking' the  $E$-vertices, $\mathcal{S}$ must have moved some of these gray tokens. 
 
 Observe that a gray token must move an even number of steps in order to return to its original vertex.  Let $\mathcal C$ be the set of $V$-vertices whose gray tokens take at least $4$ steps. We prove that $\mathcal C$ forms a vertex cover and has size at most $q$.
 
 To prove that $\mathcal C$ is a vertex cover, consider an edge $(x,y)$ of $G$. We must move four tokens of colour $xy$ from $T$ to $P$.  These tokens must come from the path descending from the $V$-vertex for $x$ or for $y$.  Suppose without loss of generality that at least two tokens of colour $xy$ come from the path descending from $x$.  We claim that the gray token initially at $x$ makes at least 4 steps and thus is in $\mathcal C$.  If the gray token exits the path from the root of $T$ to the leaf descended from $x$, then this takes at least 2 steps out and 2 steps back for a total of at least 4 steps.  If it does not exit this path, then it must make at least 1 step downward to let each of the two $xy$ coloured tokens pass, and 1 step back up for a total of at least 4 steps.
 
We now show that the size of $\mathcal{C}$ is at most $q$. Suppose for a contradiction that $\mathcal{C}$ contains $(q+1)$ or more vertices. Since each vertex in $\mathcal{C}$ corresponds to a gray token that takes at least 4 steps, the cost for $(q+1)$ such gray tokens must be at least $(q+1)4(1 + n^5)$. One can verify that $4(1+n^5) > 2(16(3n)^2+92(3n))+48n$ (for every $n\ge 5$), so for   $m \le 3n$ and $\beta = 16m^2+92m$, we obtain $(q+1)4(1 + n^5) = 4q(1 + n^5) + 4(1 + n^5)
> 4q(1 + n^5) + 2\beta+48q = 2\alpha+2\beta$.  

The membership in NP comes from the observation that any solution with the minimum number of swaps requires only a quadratic number of swaps. Hence a certificate (i.e., a set of token swap operations) can be verified in quadratic time by executing the given swaps and checking whether the total cost is within the cost bound.
\end{proof}

\section{Weighted Coloured Token Swapping}
\label{sec:weight-colour}
In this section we give polynomial-time algorithms for weighted coloured token swapping on paths and stars.
Recall from the previous section our convention that we have coloured tokens and coloured vertices, with one token at each vertex, and with the number of tokens of each colour equal to the number of vertices of that colour.  
The goal is to perform swaps to get each token to a vertex that matches its colour.
Each colour $c$ has a weight $w(c)$ and the cost (or weight) of performing a swap on two tokens of colour $c$ and $c'$ is $w(c) + w(c')$. The objective is to minimize the total cost (weight) of the swaps.
Note that standard token swapping is the special case where all the colours are distinct and all the weights are $\frac{1}{2}$, since each swap moves two tokens.

A main result we will prove is that for paths and stars, colours and weights can be dealt with separately. 
More precisely, we prove that the following approach is correct. First we ignore the weights and find the optimum assignment of tokens to destination vertices (of the correct colour).  
For this, we use the known algorithms for coloured token swapping on paths~\cite{AmirP15,bonnet2017complexity}
and on stars~\cite{bonnet2017complexity}.
After that, we solve a standard weighted token swapping problem.

We note that such a separation of colours and weights does not hold for trees in 
general, as the NP-hardness proof in the previous section shows. 
However, when the number of colours is 2, the weights and colours do separate---we should never swap two tokens of the same colour, and therefore every swap costs $w(c_1) + w(c_2)$ where $c_1$ and $c_2$ are the two colours.  This means that, for 2 colours, weighted coloured token swapping is no harder than coloured token swapping.  
Yamanaka et al.~\cite{yamanaka2018colored} gave 
a polynomial-time algorithm for 2-coloured token swapping on general graphs. Thus, weighted 2-coloured token swapping can also be solved in polynomial time.

Our main result in this section is an algorithm for weighted coloured token swapping on stars.  Before that, we give a brief solution for paths.

\subsection{Weighted coloured token swapping on  paths}

Coloured token swapping on a path can be solved in polynomial time~\cite{AmirP15,bonnet2017complexity}.  This can be seen as follows. 
As mentioned above, we should never swap two tokens of the same colour.   As noted independently in several sources~\cite{yamanaka2018colored,AmirP15,bonnet2017complexity}, for the case of paths, this constraint imposes a unique assignment of tokens to vertices: the $i^{\rm th}$ token of colour $c$ along the path must be assigned to the $i^{\rm th}$ vertex of colour $c$.  

It remains to solve the weighted token swapping problem on paths. 
As in the unweighted case, the required swaps correspond precisely to the inversions, i.e., the pairs of tokens $t, t'$ whose order in the initial token placement differs from their order in the final token placement. The minimum weight of a swap sequence is the sum, over all inversions $t, t'$ of $w(t) + w(t')$.


\subsection{Weighted coloured token swapping on stars}
\label{sec:WC-stars}

In this section we give a polynomial time algorithm for the weighted coloured token swapping problem on a star. 
As announced above, we will show that weights and colours can be dealt with separately.  We use the algorithm for coloured token swapping on a star given by Bonnet et al.~\cite{bonnet2017complexity}.
 
\subsubsection{Weighted token swapping on a star}

In this subsection we assume that every token has a distinct colour so we know exactly which vertex every token must move to. 
Each token $t$ has a weight $w(t)$ and the cost of swapping tokens $t$ and $t'$ is $w(t) + w(t')$.
Let $H$ and $U$ be the sets of tokens initially on the happy and unhappy leaves, respectively.
 Let $A$, the set of \emph{active} tokens, be all tokens except those in $H$, i.e., $A$ is $U$ plus the token at the center vertex.

In the token permutation, the cycle that contains the token at the center vertex of the star will be called the 
\emph{unlocked cycle}, and all other cycles will be called \emph{locked cycles}.
Using this terminology, Lemma~\ref{lemma:stars} states that the optimum number of swaps to solve the unweighted token swapping problem is $n_U + \ell$, where $n_U = |U|$ and $\ell$ is the number of non-trivial locked 
cycles. The intuition for the lemma, and the reason for our terminology, 
is that every locked cycle must be `unlocked' by an external token,
introducing one extra swap per locked cycle. 

The number of swaps performed in the weighted case must be at least $n_U + \ell$ and we will show that an optimum solution uses either this lower bound or two extra swaps.  
The idea is the following: each of the
locked cycles must be unlocked by some other token, and we want to use
the cheapest possible token for this. 
Either we will use an active token and perform $n_U + \ell$ swaps, 
or we will introduce two extra swaps  
that bring and return a globally
cheapest token from an initially happy leaf to the star center and
use this token to unlock all the locked cycles.

\smallskip\noindent
{\bf Notation.} The following notation will be used throughout Section~\ref{sec:weight-colour}.
 
Let $X$ be the unlocked cycle.
Let $x$ be a minimum weight token in $X$, $a$ be a minimum weight token in $A$, and $h$ be a minimum weight token in $H$ ($h$ might not exist if there are no happy leaves).
Observe that $w(a) \le w(x)$. 
As above, let $\ell$ denote the number of non-trivial locked cycles in the input token permutation.
Finally, let $d(t)$ be the distance of token $t$ from its home
and let $D_w = \sum_{{\rm token}\;t}{w(t)d(t)}$.  
Observe that $D_w$ is a lower bound on the cost of weighted token swapping.     

Before presenting the algorithm, we give an alternative formula for $D_w$.  We will use this in the forthcoming section on weighted coloured stars.  Also, it implies that  in the case of unit weights, $D_w = 2n_U$, which 
will clarify how the present algorithm generalizes the unweighted case.
For vertex $v$, recall our notation that $s^{-1}(v)$ is the initial token at $v$, and  $f^{-1}(v)$ is the final token at $v$.  Thus, a leaf vertex $v$ is happy if and only if $s^{-1}(v) = f^{-1}(v)$.

\begin{claim}
\label{claim:alternate-D}  
%
$D_w = \sum_{v \in U} ( w(s^{-1}(v)) + w(f^{-1}(v)))$.
\end{claim}
\begin{proof} If $t$ is an unhomed token whose initial and final vertices are both leaves, then it contributes $2w(t)$ to both sides of the equation.
If $t$ is a token whose initial vertex is the center vertex and whose final vertex is a leaf, then it contributes $w(t)$ to both sides.
Similarly, a token whose initial vertex is a leaf and whose final vertex is the center, contributes $w(t)$ to both sides.
Finally, a token that is home contributes 0 to both sides.
\end{proof}

\begin{corollary}
\label{cor:unit-weight-D}
When the weights are all 1, $D_w = 2n_U$.
\end{corollary}

We now describe the algorithm for weighted token swapping on a star. The algorithm uses the best of three possible strategies, all of which begin the same way:

\begin{description}
\item{Strategy 1.}  Begin solving the unlocked cycle $X$ by repeatedly swapping the token from the star center to its home until
the token $x$ is on the star center. 
Next, use $x$ to unlock and solve all the locked cycles. Finally, complete solving $X$.
The total weight is $D_w +2w(x)\ell$.  

\item{Strategy 2.} This strategy only
applies when  $w(a) < w(x)$, in which case $a \in U \setminus X$.   Begin solving the unlocked cycle $X$ by repeatedly swapping the token from the star center to its home until
the token $x$ is on the star center.  
Then swap $x$ with $a$.  Suppose $a$ was in the locked cycle $L$.  
Use $a$ to unlock and solve all the other locked cycles, leaving tokens of $X$ and $L \setminus \{a \}$ fixed.
Then use $a$ to solve cycle $L$, which will return $x$ to the center token.
Finally, complete solving $X$.
The effect is that one locked cycle is unlocked by $x$ at a cost of $2w(x)$ and $\ell -1$ cycles are unlocked by $a$ at a cost of $2w(a)(\ell -1)$, for 
a total cost of $D_w +2w(x) + 2w(a)(\ell-1)$.

\item Strategy 3.  This strategy only applies when $h$ exists.
Begin solving the unlocked cycle $X$ by repeatedly swapping the token from the star center to its home until
the token $x$ is on the star center.  
Then swap $x$ with $h$.  Use $h$ to unlock and solve all the locked cycles, leaving tokens of $X$ fixed.
Then swap $h$ and $x$.  Finally, complete solving $X$.
The total weight is $D_w + 2w(x) + 2w(h) + 2w(h)\ell = D_w + 2w(x) + 2w(h)(\ell + 1).$
\end{description}

To decide between the strategies we find the minimum of $w(x)(\ell -1)$, $w(a)(\ell -1)$, $w(h)(\ell + 1)$ and use the corresponding strategy 1, 2, or 3, respectively, achieving a total weight of $D_w + 2w(x) + 2\min\{w(a)(\ell -1), w(h) (\ell + 1)\}$. 

\begin{theorem}
\label{thm:min-weight-star}
The above algorithm finds a minimum weight swap sequence and the weight of the swap sequence is:
\[
D_w + 2w(x) + 2\min\{w(a)(\ell -1), w(h) (\ell + 1)\}.
\]
\end{theorem}

Observe that in
the case of unit weights, $D_w = 2n_U$ by Corollary~\ref{cor:unit-weight-D}, so
the theorem says that the minimum number of token moves is  $ 2n_U + 2 + 2(\ell - 1) = 2n_U + 2\ell$, i.e., the number of swaps is $n_U + \ell$, which matches what we know for the unweighted case.

To prove the theorem, we will need the following result about the unweighted star.

\begin{lemma} 
\label{lem:happy-star} 
Any swap sequence on an unweighted star that moves
a happy leaf does at least two more swaps than an optimal swap
sequence.
\end{lemma}
\begin{proof}
By Lemma~\ref{lemma:stars}, solving the unweighted problem on a star optimally takes $n_U + \ell$ swaps, where $n_U$ is the number of unhappy leaves and $\ell$ is the number of non-trivial locked cycles. It suffices to check that after swapping a happy token with the center token the value given by the formula is increased by one. Indeed, the number of non-trivial locked cycles stays the same and the number of unhomed leaves increases by one, hence, the net change is +1.
\end{proof}

We now prove Theorem~\ref{thm:min-weight-star}.

\begin{proof}
The swap sequence found by the algorithm realizes the formula given in the theorem.
It remains to show that the formula provides a lower bound on the weight of any swap sequence. 

To reach its home, each token $t$ must contribute weight at least $w(t)d(t)$, for a total over all tokens of $D_w$.
If $\ell = 0$ then $U - X$ is empty so
$w(a) = w(x)$ and  
the formula evaluates to $D_w$, which is a lower bound.  Assume from now on that $\ell \ge 1$.
In addition to the moves accounted for in $D_w$, there must be at least $2 \ell$ other token moves, two for each locked cycle.
Furthermore, there must be a first move that swaps some token $t$ of $X$ with a token outside $X$. This swap can only happen when $t$ is at the center vertex.  Since $t$ will then be unhomed, there must be a move that returns token $t$ back to the center vertex.  The minimum weight for each of these moves is $w(x)$, and this provides the term $2w(x)$ in the lower bound.  Subtracting these two moves from the required $2 \ell$ moves leaves  
$2 (\ell -1)$ moves still to be accounted for.

We now consider two cases depending whether a token of $H$ is moved or not.
If no token of $H$ is moved in the swap sequence, then the best we can do for the remaining $2(\ell -1)$ moves is to use a minimum weight token 
from $A$, so the weight is at least $D_w + 2w(x) + 2(\ell -1)w(a)$.

Next, consider swap sequences that move a token of $H$.  
By Lemma~\ref{lem:happy-star}, the sequence must do at least two extra swaps, i.e., at least 4 extra token moves.  Thus the number of moves (beyond those for $D_w$)
is at least $2\ell + 4$. As argued above, we need two moves for a token of $X$, costing at least $w(x)$ each.  We also need two moves for a token of $H$ (to leave its home and then return) costing at least $w(h)$ each.  This leaves $2 \ell$ further moves.  
This gives weight at least $D_w + 2w(x) + 2w(h) + 2\ell \min \{w(a), w(h) \}$. 
If $w(h) \ge w(a)$ then this lower bound is higher than the previous ones and becomes irrelevant.  If $w(h)< w(a)$, then the bound becomes  
$D_w + 2w(x) + 2w(h) + 2\ell w(h) = D_w + 2w(x) + 2w(h)(\ell +1)$.

Thus, combining these possibilities, we have a lower bound of $D_w + 2w(x) + 2\min\{w(a)(\ell -1), w(h) (\ell + 1)\}$, which completes the proof. 
\end{proof}

\subsubsection{Coloured token swapping on a star} 
\label{coloured ts on star}

Recall that in the coloured token swapping problem, tokens and vertices
are assigned colours, possibly with multiple tokens and vertices sharing
the same colour, and the aim is to move the tokens to vertices of
corresponding colours using a minimum number of swaps.

Bonnet et al.~\cite{bonnet2017complexity} gave a polynomial time algorithm for coloured token swapping on a star.  We use their algorithm, but in order to prove (in Section~\ref{sec:weights-colours-star}) that colours and weights separate  we need some stronger properties of the algorithm, so we present the details.

The algorithm finds a \emph{token-vertex assignment} that maps each token $t$ of colour $c$ to a vertex $v$ of colour $c$, with the interpretation that token $t$ should move to vertex $v$ in the token swapping problem.  
Such an assignment yields a standard token-swapping problem which, by Lemma~\ref{lemma:stars}, requires $n_{U}+\ell$ swaps, where $n_{U}$ is the number of unhappy leaves
and $\ell$ is the number of non-trivial locked cycles.
Thus, we want a token-vertex assignment that minimizes $n_{U}+\ell$.  Note that minimizing $n_U$ is the same as maximizing $n_H$, the number of happy leaves
since $n_U = n-1-n_H$, where $n$ is the number of vertices in the star.

An optimum token-vertex assignment is found using an auxiliary multigraph 
$G$ that has a vertex for each colour and one edge for each vertex of the star: for a vertex of colour $c$ with an initial token of colour $d$, add a directed edge from $c$ to $d$ in $G$.  In case $c=d$ this edge is a loop.
See Figure~\ref{fig:coloured-star} for an example.

\begin{figure}[t]
    \centering
    \includegraphics[width=\textwidth]{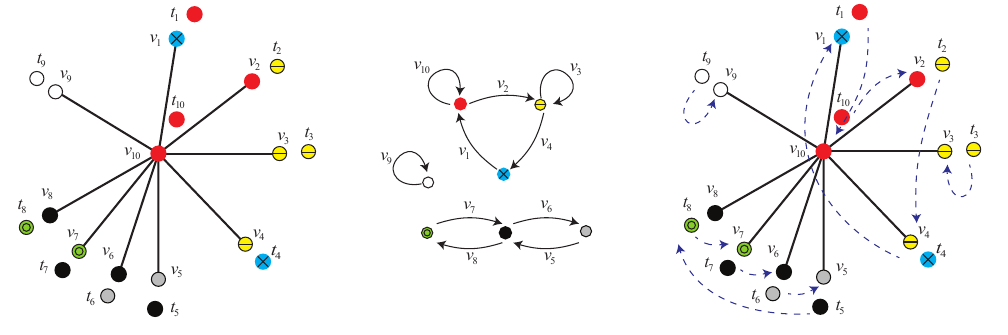}
    \caption{Left: an input for coloured token swapping on a star. The token at a vertex is drawn as a disc near the vertex.  A token must move to a vertex of the same colour.  Middle: the multi-graph $G$ with edges labelled by the corresponding vertex of the star.  There are 3 loops but one of them corresponds to the center vertex of the star, so  $\lambda=2$.  There are 3 connected components, but one is trivial, and one contains the edge corresponding to the center vertex so  $\kappa=1$.  Right: a token-vertex assignment (shown by the dashed arrows) that minimizes $n_U + \ell$.   
    One may also observe that assigning token $t_{10}$ to the center vertex $v_{10}$  or token $t_7$ to vertex $v_8$ are both sub-optimal. }
    \label{fig:coloured-star}
\end{figure}

Let $\lambda$ be the number of \emph{leaf loops} of $G$---loops corresponding to leaves of the star.  Any leaf loop corresponding to leaf vertex $v$ and token $t$ can be turned into a happy leaf by assigning token $t$ to vertex $v$.  This maximizes the number of happy leaves, $n_H$.

In the input, the number of vertices of colour $c$ is equal to the number of tokens of colour $c$.  Thus, each vertex of $G$ has in-degree equal to out-degree, which implies that any connected component in $G$ is strongly connected and has a directed Eulerian tour.
We call a connected component \emph{trivial} if it has one vertex.
Let $\kappa$ be the number of non-trivial connected components of $G$ not counting the component that contains the edge corresponding to the center vertex of the star.  

The algorithm for coloured token swapping on a star is as follows:
\begin{enumerate}
\item Find a token-vertex assignment:
\begin{enumerate}
\item Construct the multigraph $G$.
\item For each of the $\lambda$ leaf-loops, assign its token to its vertex. 
\item Remove the leaf-loops from $G$ to obtain $G'$. Observe that $\kappa$ is unchanged, and $G'$ is still Eulerian.
For each connected component of $G'$
find an Eulerian tour that traverses all the edges of the component.  
Convert each Eulerian tour to a token-vertex assignment as follows:
 Suppose the edges of the tour are labelled by vertices $v_1, v_2, \ldots, v_b$ (we are freely re-labelling vertices to ease the notation), and suppose that the edge of $G$ labelled $v_i$ goes from colour $c_{i-1}$ to colour $c_i$ (subscript addition modulo $b$). Then vertex $v_i$ has colour $c_{i-1}$ and the colour of its initial token, say $t_i$, is $c_i$.  The next edge in the tour  corresponds to vertex $v_{i+1}$ of colour $c_i$.  
Assign token $t_i$ to vertex $v_{i+1}$.  Note that both have colour $c_i$.  This assignment is well-defined since the edges of the walk correspond to distinct vertices with distinct initial tokens.
Note that this token-vertex assignment introduces a cycle 
$t_1, t_2, \ldots t_b$ in the corresponding token permutation.
 \end{enumerate}
\item Solve the (un-coloured) token swapping problem determined by the computed  token- vertex assignment.
\end{enumerate}
 
This algorithm produces a token-vertex assignment with $\lambda$ happy leaves, and $\kappa$ non-trivial locked cycles, one for each non-trivial connected component of $G'$
except the component that contains the edge corresponding to the center vertex of the star. In other words, $n_H = \lambda$, $n_U = n-1-\lambda$ and $\ell = \kappa$.  
Thus the number of swaps is $(n-1-\lambda) + \kappa$ by Lemma~\ref{lemma:stars}.

 The following lemma states two properties that are used to prove that the algorithm is correct, and a third property that will be needed in the following subsection when we consider weights.

\begin{lemma} 
\label{lem:token-assign-properties}
Any token-vertex assignment $T$ has the following properties:
\begin{enumerate}
    \item $T$ has at most $\lambda$ happy leaves.
    \item $T$ has at least $\kappa$ non-trivial locked cycles.
    \item If $T$ has $\lambda$ happy leaves then the tokens in the unlocked cycle of $T$ are a subset of $X_A$, where $X_A$ is the set of tokens that are in the unlocked cycle resulting from the above algorithm. 
\end{enumerate}
\end{lemma}
\begin{proof}
1. Happy leaves only arise from leaf loops so $T$ has at most $\lambda$ happy leaves.

2. The token permutation corresponding to $T$ can 
be expressed as a set $\cal C$ of cycles.
We claim that each cycle $C \in {\cal C}$ corresponds to a closed walk $\bar C$ of the same size in $G$ and that every edge of $G$ is in $\bar C$ for some $C \in {\cal C}$.
This will prove property 2, because it 
implies that we need at least one cycle for each connected component in $G$, and more precisely, that we need at least one non-trivial locked cycle for each of the components counted in $\kappa$.

Consider an edge of $G$, say the edge corresponding to the vertex whose initial token is $t_1$.
Token $t_1$ appears in some cycle $C \in {\cal C}$, say $(t_1, t_2, \ldots, t_b )$.
(We are freely re-naming tokens, vertices, and colours in this proof.)
Suppose token $t_i$ has colour $c_i$ and is initially at vertex $v_i$.  
Then the cycle moves token $t_i$ to vertex $v_{i+1}$ (subscript addition modulo $b$). Since the token-vertex assignment respects the colours, vertex $v_{i+1}$ has colour $c_i$.  
Also, vertex $v_{i+1}$ has initial token $t_{i+1}$ of colour $c_{i+1}$.  Thus there is a corresponding edge $c_i, c_{i+1}$ in $G$.  Therefore, the cycle corresponds to a closed walk in $G$.  Also, this closed walk uses the edge we began with, the one whose initial token is $t_1$.

3. The unlocked cycle of $T$ is the one that contains the token $t_c$ initially on the center vertex $u$.  By the argument above, the tokens in the unlocked cycle must come from the connected component of $G$ that contains the edge labelled with $u$.  This set of tokens consists of $X_A$ together with some tokens of leaf-loops.  But if $T$ has $\lambda$ happy leaves, then all the leaf-loops have been turned into happy leaves, so the set of tokens is reduced to $X_A$.  Thus, the tokens of the unlocked cycle are a subset of $X_A$.   
\end{proof}

 From this lemma the correctness of the algorithm follows: 
 
\begin{theorem}[\cite{bonnet2017complexity}]
The above algorithm uses $(n-1-\lambda) + \kappa$ swaps and this is the minimum possible.
\end{theorem}
\begin{proof}
As already stated, the algorithm uses $(n-1-\lambda) + \kappa$ swaps.

By Lemma~\ref{lem:token-assign-properties} any other token-vertex assignment results in at most 
$\lambda$ happy leaves, i.e.~at least $n-1-\lambda$ unhappy leaves, and at least $\kappa$ non-trivial locked cycles, and therefore, by Lemma~\ref{lemma:stars}, at least $(n-1-\lambda) + \kappa$ swaps.  
\end{proof}

\subsubsection{Weighted coloured token swapping on a star}
\label{sec:weights-colours-star}

Our algorithm for weighted coloured token swapping on a star is as follows:

\begin{enumerate}
\item Ignore the weights and find a token-vertex assignment as in Step 1 of the algorithm in the previous section.
\item Using this token-vertex assignment $T$ and the original token weights,
run the algorithm for the (uncoloured) weighted star.
\end{enumerate}

 In order to show that this algorithm is correct, we will first show that any optimum token-vertex assignment must turn all leaf-loops into happy leaves.  
After that 
we only need to compare the solution found by the algorithm to solutions with this property. 

\begin{claim}
\label{claim:coloured-weighted-star-claim} 
Suppose $T$ is a token-vertex assignment and there is a leaf-loop consisting of a leaf vertex $v$ with token $t$ such that both $v$ and $t$ have colour $c$, but the token-vertex assignment does not assign $t$ to $v$.  Then $T$ is not optimum for the weighted problem.
\end{claim}

\begin{proof}
By Theorem~\ref{thm:min-weight-star}, the cost of $T$ is 
$$F(T) = D_w + 2w(x) + 2\min\{w(a)(\ell -1), w(h) (\ell + 1)\},$$ 
where $D_w, w(x), w(a), w(h)$, and $\ell$ depend on $T$.  We will construct a new token-vertex assignment $T'$ that assigns $t$ to $v$ and has $F(T') < F(T)$.

Since $t$ is not assigned to $v$, $t$ must be part of some non-trivial cycle $C$ in the token permutation determined by $T$. 
Suppose that the cycle $C$ contains tokens $p, t, q$ in that order (possibly $p = q$), with initial vertices $s(p), s(t){=}v,s(q)$, respectively.   
Define a new token-vertex assignment $T'$ that 
assigns $t$ to $v$, i.e., $v$ becomes a happy leaf, and shortcuts the rest of $C$ by assigning token $p$ to vertex $s(q)$.  This is valid because token $p$ and vertex $s(q)$ both have colour $c$, the same as $t$.  The new cycle $C'$ is formed by deleting $t$ from $C$.
We will compare $F(T)$ and $F(T')$ by looking at the quantities $D_w, w(x), w(a), w(h)$ and $\ell$.

First of all, no leaf becomes unhappy, so no token leaves $H$ and  
$w(h)$ does not increase.  Furthermore, $v$ becomes happy, so 
by Claim~\ref{claim:alternate-D}, $D_w$ decreases by at least $2w(t)$.

Next we show that $w(x)$ does not increase.  That would only happen if $t$ leaves the set $X$.  Then $C$ must be the unlocked cycle.  Since $t$ is at a leaf vertex, the token from the center vertex remains in $C'$, so $C'$ is the new unlocked cycle.  Furthermore, 
token $p$, which is a `twin' of $t$ in the sense that it has the same colour and weight, remains in $C'$, so $w(x)$ remains the same.

Finally we must consider $\ell$ and $w(a)$.  Here we will separate out one special case---when $|C| = 2$ and $C$ exchanges two leaf tokens, in which case $C'$ becomes a trivial locked cycle.  If we are not in the special case then either $C'$ is a non-trivial locked cycle, or $C'$ is the unlocked cycle. In  either case $C$ has the same status, so $\ell$ is unchanged and $t$'s twin $p$ remains in the active set $A$ so $w(a)$ does not increase. Thus $F(T') < F(T)$ when we are not in the special case.

It remains to consider the special case when $C$ exchanges two leaf tokens.
Then $C$ was a non-trivial locked cycle, but $C'$ is a trivial locked cycle.  Thus $\ell$ decreases by 1.  Furthermore, by Claim~\ref{claim:alternate-D}, $D_w$ decreases by at least $4w(t)$ since two leaves become happy.
If $w(a)$ does not increase then we are fine. If it does increase then $w(a) = w(t)$ and $t$ was the minimum weight element in $A$.
Because we are in the special case, both $t$ and its twin token $p$ have left $A$ and joined $H$.
If $F(T)$ is determined by $w(h)(\ell+1)$ we are again fine. Hence we only need to provide an additional argument if $F(T) = D_w + 2w(x) + 2w(a)(\ell -1)$. Since we now have a token of weight $w(a)=w(t)$ in $H$, Strategy 3 gives a swap sequence for $T'$ of weight at most $(D_w - 4w(t)) + 2w(x) + 2w(t)\ell = D_w - 2w(t) + 2w(x) + 2w(t)(\ell -1)$.  Thus $F(T') < F(T)$ even in the special case.
\end{proof}

With this claim in hand, we are ready to prove that the algorithm is correct.

\begin{theorem}
The above algorithm solves the weighted coloured token swapping problem
on a star optimally.
\end{theorem}

\begin{proof} 
By Theorem~\ref{thm:min-weight-star}, the cost of a token-vertex assignment $T$ is 
$$F(T) = D_w + 2w(x) + 2\min\{w(a)(\ell -1), w(h) (\ell + 1)\},$$ 
where $D_w, w(x), w(a), w(h)$, and $\ell$ depend on $T$.  

We will compare the cost of a token-vertex assignment $T_A$ found by the algorithm to an optimum token-vertex assignment $T_{\rm OPT}$.  
By Claim~\ref{claim:coloured-weighted-star-claim}, $T_{\rm OPT}$ turns all leaf-loops into happy leaves, so it has $\lambda$ happy leaves.  The algorithm does the same, so $T_A$ and $T_{\rm OPT}$ have the same set $H$ of tokens on happy leaves, and the same set $U$ of tokens on unhappy leaves.  This implies that $w(a)$ and $w(h)$ are the same for $T_A$ and $T_{\rm OPT}$.  

Next, we claim that $D_w$ is the same for $T_A$ and $T_{\rm OPT}$. This follows directly from Claim~\ref{claim:alternate-D} since the set of unhappy leaves is the same. 

It remains to compare $\ell$ (the number of non-trivial locked cycles) and $w(x)$ between $T_A$ and $T_{\rm OPT}$.  
Both values should be as small as possible in $T_{\rm OPT}$.  
The algorithm achieves $\ell = \kappa$ and $w(x) = \min\{w(t) : t \in X_A\}$, where $X_A$ is the set of tokens in the unlocked cycle of $T_A$. 
By Lemma~\ref{lem:token-assign-properties}(2) $T_{\rm OPT}$ has at least $\kappa$ non-trivial locked cycles.  
By Lemma~\ref{lem:token-assign-properties}(3), $T_{\rm OPT}$'s set of tokens in the unlocked cycle is a subset of $X_A$ (here we again use the fact that $T_{\rm OPT}$ has $\lambda$ happy leaves). Thus $T_A$ and  $T_{\rm OPT}$ achieve the same values for $\ell$ and $w(x)$.
This completes the proof that the algorithm achieves the minimum value of $F(T)$. 
\end{proof}

\section{Token Swapping on Brooms}
\label{sec:brooms}

In this section we give a polynomial-time algorithm for token swapping on a broom. 

A \emph{broom} is a tree that consists of a star joined to the endpoint of a path.
Suppose that the broom's star has $k$ leaves, $v_1, \ldots, v_k$ and one center vertex, $v_{k+1}$, and its path has $n-k$ vertices, $v_{k+1}, v_{k+2}, \ldots, v_{n}$.
See Figure~\ref{fig:broom}.
We will call vertices $v_1, \ldots, v_k$ the \emph{star leaves} and call vertices $v_{k+1}, v_{k+2}, \ldots, v_{n}$ the \emph{path vertices}.
We will call $v_{k+1}$ the \emph{center vertex}, though we also consider it to be a path vertex.
For purposes of identification, we will orient the edges of a broom from lower index vertices to higher index vertices, and we will draw brooms with edges directed to the right as in the figure.  

\begin{figure}[ht]
    \centering
    \includegraphics[width=\textwidth]{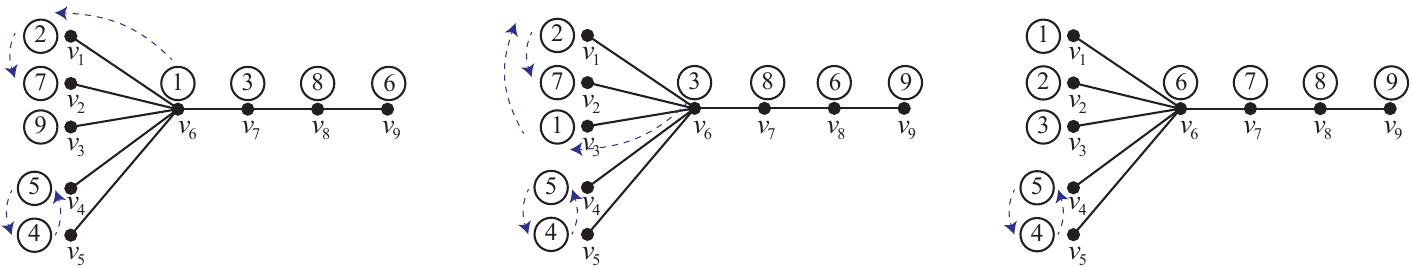}
    \caption{A broom with leaves $v_1, v_2, v_3, v_4, v_5$, path $v_6, v_7, v_8, v_9$ and center vertex $v_6$.  Left: The initial placement of tokens (drawn in circles).
    Destinations within the star are indicated by dashed lines. 
    Tokens 4 and 5 form a cycle on the star leaves and will be exchanged by Step 2 of the algorithm.  The algorithm begins with token 9.  Step 1(b) does not apply, so token 9 is moved to its home.  Middle: The token placement after token 9 is homed.  Token 8 will be homed next.  For token 7, Step 1(b) applies, since the chain starting from token 3 on the center vertex leads to token 7. Right: The token placement after tokens 8 and 7 are homed.   
    }
    \label{fig:broom}
\end{figure}

The input to the problem is a broom and an initial assignment of tokens $1, 2, \ldots, n$ to the vertices.  
The problem is to find a minimum length sequence of swaps to ``sort'' the labels, i.e., to get label $i$ at vertex $v_i$ for all $i$.

In 1999 Vaughan~\cite{vaughan1999broom} gave a polynomial-time algorithm for token swapping on a broom. 
Her algorithm is a bit more complicated than ours and her proof of correctness is long (14 pages).
 Another polyomial-time algorithm for token swapping on a broom was given by Kawahara et al.~\cite{kawahara2017time} (with proof in~\cite{kawahara2016time}).  They use the term ``star-path'' instead of ``broom'' and were unaware of Vaughan's work.  Their algorithm is simple, but their proof is again long (11 pages).  Our correctness proof is shorter.
 Further comparisons can be found below.
 It is interesting that three independently developed algorithms for the problem are similar in nature, and that all three correctness proofs are non-trivial.  None of the three manage to give a formula for the number of swaps that is independent of the operation of algorithm (such as the formulas for paths and stars).

We will call the tokens $1, \ldots, k$---the tokens that want to end up at the star leaves---the \emph{star tokens} and we will call the other tokens the \emph{path tokens}.   A \emph{centered star chain} of token $t_{m+1}$ is a sequence of tokens $t_1, \ldots, t_{m+1}$ such that 
\begin{itemize}
    \item $t_1$ is currently at the center vertex,
    \item $t_2, \ldots, t_{m+1}$ are currently at star vertices,
    \item for $1 \le i \le m$, $t_i$'s home currently contains $t_{i+1}$, i.e., the sequence forms a chain in the token permutation.
\end{itemize} 
For example, in Figure~\ref{fig:broom}(Middle) token 7 has a centered star chain $3,1,2,7$.

\subsection{Algorithm for token swapping on a broom}
Our algorithm $\cal A$ is as follows. 
\begin{enumerate}
\item While there is a path token that is not home:
\begin{enumerate}
\item Let $p_{\max}$ be the maximum path token that is not home.
\item If $p_{\max}$ is on a star leaf and has a centered star chain  
$t_1, \ldots, t_m, p_{\max}$ then perform the $m$ swaps that move $t_i$ home for $i = 1, \ldots, m$.
The final swap moves $p_{\max}$ to the center vertex.   
\item 
Home $p_{\max}$. 
\end{enumerate}
\item (At this point all path tokens are home.) Solve the star.
\end{enumerate}

Step 2 of the algorithm just involves solving the star at the end, and is well-understood from previous work. 
Note that the algorithm does not move happy leaves---thus, correctness of the algorithm implies that the Happy Leaf Conjecture is true for brooms.

We begin our analysis of the algorithm by noting 
that Step 1 has the following nice properties:
\begin{itemize}
    \item[(P1)] Every path token moves left for a time (as larger tokens move past it) and then moves right to its home.  
    \item[(P2)] Every swap performed by the algorithm either moves the largest unhomed path token to the right, or homes a star token.  When we are lucky, a single swap does both (the last swap performed in Step 1(b)).   
\end{itemize}

We prove correctness of the algorithm in Section~\ref{sec:broom-correct} below.  Before that, we analyze the number of swaps used by the algorithm (this is not needed for the correctness proof) and we compare our algorithm to the others.

\subsection{Algorithm analysis}
We first consider Step 2.  Let $\ell_S$ be the number of non-trivial cycles in the initial permutation that only involve star tokens on star leaves, and let $n_S$ be the number of tokens that are in these cycles.  These tokens are not touched by Step 1, and we claim that, conversely, after Step 1, these are the only unhomed tokens. 
The reason is that no swap performed by the algorithm can create a cycle involving star tokens on star leaves---this is precisely what Step 1(b) prevents. 
Thus, by Lemma~\ref{lemma:stars} the number of swaps performed by Step 2 of the algorithm is $n_S + \ell_S$.

It remains to analyze the number of swaps performed in Step 1.  To ease notation, we now assume that $n_S$ is 0.
Let $W$ be the number of swaps performed by Step 1 of the algorithm.
Write $W = W_P + W_S$ where $W_P$ is the number of swaps that involve a path token and $W_S$ is the number of swaps that involve two star tokens. 
Then $W_S = S_U - L$, where $S_U$ is the number of unhomed star tokens 
and $L$ is the number of `lucky' swaps that move $p_{\max}$ to the right AND home a star token.  Equivalently, $L$ is the number of times Step 1(b) of the algorithm is performed.

To analyze $W_P$, we will allocate each of its swaps to the maximum path token involved in the swap.
Then $W_P = \sum_p W(p)$ where $W(p)$ is the number of swaps allocated to path token $p$. 
For any path token $p$, 
define $d(p)$ to be the distance (the number of edges) from a star leaf to $p$'s home.  Observe that it does not matter which star leaf we use in this definition.
Relative to the initial token placement, let $r(p)$ be the number of tokens smaller than $p$ that occur to the right of $p$.  

\begin{claim}
    $W(p)$ is the minimum of $d(p)$ and $r(p)$.
\end{claim}

\begin{proof}
By Property (P1), the movement of every path token $p$ during the algorithm consists of some number of steps to the left followed by some number of steps to the right until $p$ reaches its home.  The leftward steps are caused by larger tokens moving over $p$ and are allocated to the larger tokens.  $W(p)$ counts the rightward steps of token $p$.

Let $C$ be the configuration of tokens after the algorithm has homed all tokens larger than $p$.  In $C$, the tokens larger than $p$ take up all the vertices to the right of $p$'s home, so $p$ will be to the left of its home.  Let $r_C(p)$ be the number of tokens smaller than $p$ that occur to the right of $p$ in configuration $C$. Then $W(p)= r_C(p)$.  No swap has occurred between $p$ and a smaller token, so $r_C(p) \le r(p)$. Furthermore, $r_C(p)=r(p)$ unless $p$ is at a star leaf (in which case some of the original $r(p)$ vertices may have moved to other star leaves). 
We now consider two possibilities:  
\begin{itemize}
\item If $p$ is at a star leaf in configuration $C$ then its distance from its home is $d(p)$ so $W(p) = d(p) = r_C(p)$,  which implies that $d(p) \le r(p)$.  Thus $W(p) = \min\{d(p),r(p)\}$. 
\item If $p$ is not at a star leaf in configuration $C$ then $r(p) = r_C(p) <d(p)$.  Then $W(p)= r(p)$, and again $W(p) = \min\{d(p),r(p)\}$.
\end{itemize}

In either case, $W(p) = \min\{d(p),r(p)\}$.
\end{proof}

In summary, the number of swaps, $W$, performed by Step 1 of the algorithm can be expressed in terms of the initial token placement except for the 
`lucky' term $L$, which is found by running the algorithm:

\begin{lemma}
$W = \sum \{\min\{d(p),r(p)\} : p \text{ a path token} \} + S_U -L $.
\end{lemma}

\paragraph*{Vaughan's algorithm.}  The algorithm by Vaughan~\cite{vaughan1999broom} is similar to our algorithm.
We introduce one new concept, generalizing a centered star chain.  A \emph{star chain}
of token $t_{m+1}$ is a sequence of tokens $t_1, \ldots, t_{m+1}$ such that 
\begin{itemize}
    \item $t_1$ is not at a star vertex,
    \item $t_2, \ldots, t_{m+1}$ are currently at star vertices,
    \item for $1 \le i \le m$, $t_i$'s home currently contains $t_{i+1}$, i.e., the sequence forms a chain in the token permutation.
\end{itemize}

Using this terminology, Vaughan's algorithm is as follows.
\begin{enumerate}
    \item While there is a path token on a star leaf:
    \begin{enumerate}
        \item Let $p_{\max}$ be the maximum path token on a star leaf.
        \item Let $t$ be the leftmost token on a path vertex that is smaller than $p_{\max}$.
        \item Swap $t$ leftward along the path to the center vertex.  Let $t_1, \ldots, t_{m}, p_{\max}$ be the star chain of $p_{\max}$.   
        Swap $t$ to the star vertex that is the home of $t_1$ and continue resolving the star chain (next moving $t_2$) until $p_{\max}$ is swapped to the center vertex.  Note that token $t$ will be home if and only if the star chain was centered. 
    \end{enumerate}
    \item Solve the star.  Solve the path.
\end{enumerate}

It can be shown that the two algorithms are equivalent.  
Vaughan proved that her algorithm is correct by transforming any optimal swap sequence to one that matches her algorithm in its first phase.  
Her proof is difficult because properties (P1) and (P2) do not hold.

\paragraph*{Alternate algorithm.}  The algorithm by Kawahara et al.~\cite{kawahara2016time} is very simple:  

\begin{enumerate}
\item While there is a token that is not home:
\begin{enumerate}
\item Let $t_{\max}$ be the maximum token that is not home.
\item  While the center vertex contains a star token, home it.
\item Home $t_{\max}$. 
\end{enumerate}
\end{enumerate}

Their algorithm satisfies (P2) but not (P1) (in particular, a path token may bounce around inside the star) which makes their proof difficult.

\subsection{Correctness}
\label{sec:broom-correct}
In order to prove that algorithm $\cal A$ finds an optimal swap sequence, i.e., a swap sequence with a minimum number of swaps,
we will prove that for any input to the token swapping problem on a broom  
there is an optimal swap sequence with some nice properties, and then apply induction on $n$.
We begin with an easy observation about optimal swap sequences.  

\begin{claim}
\label{claim:no-duplicate-swap}
An optimal swap sequence does not contain two swaps that swap the same two tokens.
\end{claim}
\begin{proof}
We prove the contrapositive.  Suppose the swap sequence $\sigma = \sigma_1 \sigma_2 \cdots \sigma_m$ sorts the tokens but contains two swaps $\sigma_i$ and $\sigma_j$, $i<j$, that swap the same two tokens $a$ and $b$.  Modify $\sigma$ by deleting $\sigma_i$ and $\sigma_j$ and, for each $k$, $i < k < j$, exchanging the roles of $a$ and $b$.  The result is a shorter swap sequence that sorts the tokens.   
\end{proof}

Next we show that an optimal swap sequence that minimizes the number of swaps on path edges has some nice properties. 
For convenience, we will say that a swap sequence is \emph{P-optimal} if it is an optimal swap sequence and minimizes the number of swaps on path edges.  

\begin{lemma}
\label{lemma:minimal-properties}
A P-optimal swap sequence $\sigma$ has the following properties:
\begin{enumerate}
    \item \label{property:no-star-swap} No swap of two star tokens occurs on a path edge.
    \item \label{property:larger-right} For every swap on a path edge that involves a path token, the larger of the two tokens moves to the right.
    \item \label{property:no-re-enter} No token crosses the first path edge, $(v_{k+1},v_{k+2})$, from left to right and later from right to left.  (No token ``re-enters the star''.)  In particular, this implies that no star token crosses the first path edge from left to right. 
    \item \label{property:larger-exit-star} The path tokens that cross the first path edge from left to right do so in order (the first one that crosses is larger than the others, etc.). 
\end{enumerate}
\end{lemma}
\begin{proof}

\paragraph*{(\ref{property:no-star-swap})} Suppose two star tokens $s$ and $t$ swap on a path edge in $\sigma$. We will show how to modify $\sigma$ 
to obtain an optimal swap sequence with fewer swaps on the path edges.  
Sometime after the swap of $s$ and $t$ in $\sigma$, both $s$ and $t$ must enter the star (perhaps multiple times).  
Suppose the last time one of them enters the star, it is $t$ that enters.  Then in this configuration $T$, token $t$ is on the center vertex and $s$ is at a star leaf.  
We will modify $\sigma$ as follows. 
Omit the swap of $s$ and $t$. Then carry out the swap sequence, but with $s$ and $t$ exchanged.  
Continue until the point in $\sigma$ where configuration $T$ was reached. Then swap $s$ and $t$.  Now we are back in configuration $T$.  
The same number of swaps have been performed, but a swap on a path edge has been replaced by a swap on a star edge, so the number 
of swaps on path edges has been reduced.

\paragraph*{(\ref{property:larger-right})} Call a swap ``bad'' if it swaps two tokens by moving a path token to the left and a smaller token to the right.
Suppose $\sigma$ has a bad swap on a path edge.  As above, we will show how to modify $\sigma$ to obtain an optimal swap sequence with fewer swaps on the path edges.  
Suppose that, in $\sigma$,  path token $p$ swaps to the left with a 
smaller token $t$ on a path edge.  In the final configuration, $p$ must be to the right of $t$.  They cannot swap directly by Claim~\ref{claim:no-duplicate-swap}.
So it must happen that $p$ reaches a star leaf, then $t$ reaches a star leaf, then $p$ goes to the center vertex.  Call this intermediate configuration $T'$.
We will modify $\sigma$ as follows.  Omit the swap of $p$ and $t$.  Then carry out the swap sequence, but with $p$ and $t$ exchanged.  
Continue until the point in $\sigma$ where configuration $T'$ was reached.  Then swap $p$ and $t$.  Now we are back in configuration $T'$.  
The same number of swaps have been performed, but a swap on a path edge has been replaced by a swap on a star edge, so the number 
of swaps on path edges has been reduced.

\paragraph*{(\ref{property:no-re-enter})} Suppose that in $\sigma$ there is a swap $\sigma_i$ where a token $t$ crosses the first path edge from left to right, and a later swap $\sigma_j$, $j>i$ where $t$ crosses the first path edge from right to left.  Take the minimum possible $j > i$.
Suppose $\sigma_i$ involves tokens $t$ and $a$ and $\sigma_j$ involves tokens $t$ and $b$.  
We first claim that token $b$ must be on a star leaf when $\sigma_i$ takes place.  Suppose $b$ were on a path edge when $\sigma_i$ takes place. Then $b$ would be to the right of $t$.  Between $\sigma_i$ and $\sigma_j$, token $t$ may move around on path edges, but not the first path edge (by our assumption on $j$).  Token $b$ cannot 
swap twice with $t$ in an optimal swap sequence, so $b$ must still be to the right of $t$ when $\sigma_j$ takes place, a contradiction.  

We will modify $\sigma$ as follows. Just before $\sigma_i$, swap $t$ (on the center vertex) and $b$ (on a star leaf).  Then carry out the swap sequence $\sigma_i, \ldots, \sigma_{j-1}$, but with $t$ and $b$ exchanged.  Then omit swap $\sigma_j$.  This gives the same result.  
The same number of swaps have been performed, but a swap on a path edge has been replaced by a swap on a star edge, so the number 
of swaps on path edges has been reduced.

\paragraph*{(\ref{property:larger-exit-star})}
Suppose a smaller path token $s$ crosses the first path edge from left to right in swap $\sigma_i$ and later on, say in swap $\sigma_j$, $j > i$, a larger path token $t$ crosses the first path edge from left to right. By property (3), token $s$ does not re-cross the first path edge. 
Then, after $t$ crosses, tokens $s$ and $t$ are in the wrong order on the path.  
Since neither token re-enters the star by (3), 
they must swap on the path during some later swap $\sigma_h$, $h>j$.
Now consider where token $t$ is when swap $\sigma_i$ occurs.  If $t$ is on the path, then in order for $t$ to exit the star in swap $\sigma_j$, it must first swap with $s$.  But then the two tokens swap twice, which contradicts Claim~\ref{claim:no-duplicate-swap}.  Therefore, we may assume that $t$ is at a star leaf when $\sigma_i$ occurs.  We will modify $\sigma$ as follows.  Swap $s$ and $t$ just before $\sigma_i$, then carry out the swap sequence with $s$ and $t$ exchanged until swap $\sigma_h$, and omit $\sigma_h$.  
This gives the same result.  
The same number of swaps have been performed, but a swap on a path edge has been replaced by a swap on a star edge, so the number 
of swaps on path edges has been reduced.  
\end{proof}

\begin{theorem}
Algorithm $\cal A$ above finds a swap sequence with a minimum  number of swaps.
\end{theorem}
\begin{proof}
By induction on the number of tokens in the path, with the base case when the path has no edges and the broom is just a star.  
For the general induction step it suffices to show that the swaps performed in the first phase of the algorithm when $p_{\max} = n$ are part of an optimal swap sequence. 
We consider three cases depending on the initial position of token $n$.

\paragraph*{Case 1.} Token $n$ is initially on a path vertex, not the center vertex.
Let $\sigma$ be a P-optimal swap sequence.
By Lemma~\ref{lemma:minimal-properties}, Property~(\ref{property:larger-right}) token $n$ only moves to the right in $\sigma$.  
We claim that we can modify $\sigma$ to do the swaps involving $n$ first.  
Let $\sigma_i$ be the swap in $\sigma$ that homes token $n$, and let $T$ be the resulting configuration of tokens.  Let $k$ be the distance from token $n$'s initial position to its home vertex $v_n$.  Then there are $k$ swaps in $\sigma_1, \ldots, \sigma_i$ that involve token $n$.
Let $\sigma'$ be the subsequence of $\sigma_1, \ldots, \sigma_i$ consisting of the $i-k$ swaps that do not involve token $n$.
Construct a swap sequence, $\tau$, that first homes token $n$ in $k$ swaps, then performs the swaps $\sigma'$, and then proceeds with swaps $\sigma_{i+1}, \ldots, \sigma_m$.  Homing token $n$ leaves the remaining tokens in the same relative order along the path and does not alter the tokens in the star.  Thus, the swaps $\sigma'$ can be performed next, and the result is the same configuration $T$ as achieved by $\sigma$.  This implies that $\tau$ is an optimal swap sequence that matches our algorithm in the first phase. 

\paragraph*{Case 2.} Token $n$ is initially on the center vertex.
Let $\sigma$ be a P-optimal swap sequence.  Sequence $\sigma$ must contain a swap, say $\sigma_j$ that swaps
token $n$ along the first path edge $e = (v_{k+1},v_{k+2})$.  
By Lemma~\ref{lemma:minimal-properties}, Property~(\ref{property:larger-exit-star}) this is the first swap that occurs along edge $e$.
If there are no swaps before $\sigma_j$ on star edges, then we can just apply the argument from Case~1---home token $n$ first, and then proceed with the swaps that do not involve token $n$.

Thus it suffices to show that there is a P-optimal swap sequence that does not do any swaps on star edges before swapping token $n$ along the first path edge.
We will adjust $\sigma$ to make this true. 
Separate the subsequence $\sigma_1, \ldots, \sigma_{j-1}$ into two subsequences: $\sigma'$, the swaps that operate on path edges, and $\sigma''$, the swaps that operate on star edges.  
Observe that we can rearrange $\sigma_1, \ldots, \sigma_{j-1}$ to perform $\sigma'$ and then $\sigma''$, since there is no interaction between the star and the path in this time interval. Observe that $\sigma'$ does not move token $n$, and that 
$\sigma''$ returns token $n$ to the center vertex. 
We will do one more modification.  After $\sigma'$, perform the swap $\sigma_j$, which swaps token $n$ with token $a$, say.  Then perform the sequence $\sigma''$ of swaps on star edges but with $a$ in place of $n$.  Clearly this gives the same result, with the same number of swaps, and the same number of swaps on path edges.
Thus there is an optimal swap sequence that matches our algorithm in the first phase.

\paragraph*{Case 3.} Token $n$ is initially on a star leaf.  Let the token on the center vertex be $t$. 
Consider a P-optimal swap sequence $\sigma$ with the further property that it does a minimum number of swaps on star edges before the first swap, $\sigma_h$, that moves token $n$ to the center vertex.
We will show that $\sigma$ performs the same swaps as our algorithm does up to the point where the algorithm moves token $n$ to the center vertex.

If the situation in Step 1(b) applies, we will say there is a \emph{complete homing cycle} from $t$ to $n$.  Using this terminology, our algorithm does a complete homing cycle if it exists, and otherwise swaps $n$ and $t$.  We must prove that $\sigma$ does the same. 
Our proof will have two steps:

(1) $\sigma$ either swaps $n$ and $t$ or does a complete homing cycle.

(2) If there exists a complete homing cycle then $\sigma$ does it.

To prove (1), suppose that $\sigma$ does not swap $n$ and $t$ and does not do a complete homing cycle. 
By Lemma~\ref{lemma:minimal-properties}, Property~(\ref{property:larger-exit-star}) applied to $\sigma$, the first swap on the first path edge involves token $n$.  Thus, the tokens on the star and the tokens on the path do not interact until after swap $\sigma_h$. 
Consider the subsequence $\sigma'$ of $\sigma_1, \ldots, \sigma_h$ that consists of swaps on star edges. They achieve a certain permutation $\cal P$ of the tokens on the star.  By hypothesis, $\sigma'$ is a minimum length swap sequence effecting the permutation $\cal P$.  
This permutation can be written as a disjoint union of cycles, and---because the token $n$ moves to the star center, which contains token $t$---one of those cycles, say $C$ has the form $(t_0 t_1 t_2 \cdots t_\ell)$ where $t_0=t$ and $t_\ell =n$.  
By the known results for token swapping on stars (Lemma~\ref{lemma:stars}), 
$\sigma'$ then effects the cycle $C$ by swapping $t_0$ and $t_1$, then $t_1$ and $t_2$, and so on, until finally $t_{\ell-1}$ is swapped with $t_\ell$.  
Since this places $t_{\ell} = n$ on the center vertex, it must complete $\sigma'$.  Also, since $\sigma$ does not swap $n$ and $t$, we have $\ell \ge 2$.

By assumption, $\sigma'$ is not a complete homing cycle.  Therefore, 
there must be some $i =0, \ldots \ell-1$ such that the swap of $t_i$ and $t_{i+1}$ does not move $t_i$ to its home.
We claim that there is a P-optimal sequence with fewer swaps on star edges before $n$ reaches the center vertex, which will be a contradiction. 
We will modify $\sigma$ as follows.  Modify the cycle $C$ by omitting $t_{i}$ (or omit $t_1$ in case $i=0$).  In the resulting token placement, $t_{i-1}$ and $t_i$ (or $t_0$ and $t_1$) are switched compared to the configuration in $\sigma$.  Both tokens are on star leaves.  Continue the swap sequence $\sigma$ beyond $\sigma_h$ until the first time when one of these two tokens is swapped (necessarily to the center vertex).  Observe that such a time must exist since token $t_i$ is not home.  At this point, one of the two tokens is at the center vertex and the other is at a star leaf.  Swap the two tokens at this point.  The result is the same, the total number of swaps is the same, and the number of swaps on path edges is the same.  Thus we have a P-optimal swap sequence with fewer swaps on star edges before $n$ reaches the center vertex. This completes the proof of (1).

To prove (2), suppose that $\sigma$ does not do a complete homing cycle.  Then, by (1), $\sigma$ swaps $n$ and $t$ and this is not a complete homing cycle.  
Suppose, for a contradiction, that 
there is a complete homing cycle of length $\ell \ge 3$. Then the swap on $n$ and $t$ creates---in the  
permutation we need to effect---a cycle $C$
of length $\ell -1$ among the star tokens that lie at star leaves.
We now prove that it takes at least $\ell$ swaps to solve this cycle.
 (This is the same argument as used for stars.)
Each token in $C$ is distance $2$ from its home, so the sum of the distances is $2(\ell -1)$ and 
we need at least $\ell -1$ swaps, since one swap can move two of the tokens closer to their homes.  Furthermore, the first swap and the last swap that operate on the tokens in the cycle move only one of the cycle tokens.  This gives a total of $\ell$ swaps. 
Together with the swap on $n$ and $t$ this is $\ell + 1$ swaps.
However, doing the complete homing cycle first would take $\ell -1$ swaps, which is better, a contradiction to $\sigma$ being optimal.  This completes the proof of (2).
\end{proof}

\section{Conclusions and Open Questions}
\label{sec:conclusions}

We have identified a previously unexplored difficulty of token swapping on trees---namely that we must decide how and when to move tokens that are at happy leaves. This difficulty does not arise for the cases where polynomial-time algorithms are known, specifically, paths, stars and brooms.

We showed that any algorithm
for token swapping on a tree that fixes tokens at happy leaves cannot achieve better than a $\frac{4}{3}$ approximation factor, and that this lower bound rises to 2 for two of the three known approximation algorithms, thus providing tight approximation factors for them.  

Furthermore, we proved that weighted coloured token swapping is NP-complete for spiders, but polynomial-time for paths and stars. 

We conclude with some open questions. 

\begin{enumerate} 

\item Find other classes of trees for which token swapping can be solved in polynomial time.

\item Find other classes of trees for which coloured weighted token swapping can be solved in polynomial time.  Currently, there is no known class of trees where coloured weighted token swapping is NP-complete but token swapping is in P. Two candidate classes are spiders, where coloured weighted token swapping is NP-complete but token swapping is of unknown complexity, and brooms, where token swapping is in P but the coloured weighted version is of unknown complexity. 
    
    \item Characterize the class of trees for which the happy leaf  property  holds for every token assignment. Certainly the tree should not have the 10-vertex tree of Figure~\ref{fig:happy-swap-counter-ex} as a subtree.  
    
    \item Is there a polynomial time algorithm for token swapping on any tree for which the happy leaf property     holds?  This may not be easy, given the difficulty of correctly solving token swapping on a broom.  
    
    \item Is there an approximation algorithm for token swapping on a tree with approximation factor better than 2?    
  
    \item The example in Figure~\ref{fig:4-3-approx-example}, which defeats all algorithms that fix happy leaves, consists of a star joined to two paths. Such a \emph{two-tailed star} is like a broom with an extra handle.  We conjecture that there is a polynomial time algorithm for token-swapping on two-tailed stars. This would be a starting point towards solving token swapping when happy leaves must be swapped. 

    \item
    For general graphs there is a 4-approximation algorithm~\cite{miltzow2016approximation} for token swapping.  Is the approximation factor 4 tight?
    
\end{enumerate}
  

\bibliographystyle{plain}
\bibliography{refs}

\end{document}